\RequirePackage{fix-cm}
\RequirePackage{amsmath}
\documentclass{svjour3}                     
\smartqed  
\usepackage{graphicx}
\usepackage{mathptmx}      
\usepackage{mathtools}
\journalname{Statistical Papers}

\usepackage{amstext, amsmath, amssymb}

\usepackage{relsize}
\usepackage{paralist}

\usepackage{array}
\usepackage{float}

\usepackage{subcaption}
\usepackage{hyperxmp}
\usepackage{hyperref}

\hypersetup{%
    bookmarks=true,         
    colorlinks=true,
    citecolor=blue,
    urlcolor=blue,
    unicode=false,          
    pdftoolbar=true,        
    pdfmenubar=true,        
    pdffitwindow=false,     
    pdfstartview={FitH},    
    pdftitle={Analytic moment and Laplace transform formulae for the quasi-stationary distribution of the Shiryaev diffusion on an interval},    
    pdfauthor={Aleksey S. Polunchenko and Andrey Pepelyshev},     
    pdfsubject={Analytic moment and Laplace transform formulae for the quasi-stationary distribution of the Shiryaev diffusion on an interval},   
    pdfcreator={Aleksey S. Polunchenko},   
    pdfproducer={MikTeX}, 
    pdfkeywords={Laplace transform, Markov diffusions, Quasi-stationarity, Shiryaev process, Special functions, Stochastic processes}, 
    pdfnewwindow=true,      
    linkcolor=red,          
    filecolor=magenta      
}

\makeatletter
\newcommand{\longdash}[1][2em]{%
  \makebox[#1]{$\m@th\smash-\mkern-7mu\cleaders\hbox{$\mkern-2mu\smash-\mkern-2mu$}\hfill\mkern-7mu\smash-$}}
\makeatother
\newcommand{\omitskip}{\kern-\arraycolsep}

\newcommand{\linefill}{
  {-}\mkern-7mu
  \cleaders\hbox{$\mkern-2mu-\mkern-2mu$}\hfill
  \mkern-7mu{-}%
}

\newcommand{\pref}[1]{(\ref{#1})}

\renewcommand{\Pr}{\mathbb{P}} 
\DeclareMathOperator{\EV}{\mathbb{E}} 
\DeclareMathOperator{\Var}{\mathrm{Var}} 

\newcommand{\abs}[1]{\left\vert#1\right\vert}

\usepackage{bbm}
\newcommand{\indicator}[1]{{\mathbbm{1}_{\left\{#1\right\}}}}

\graphicspath{{./gfx/}}

\usepackage[norefs,nocites]{refcheck}
\usepackage{silence}
\begin{document}

\title{Analytic moment and Laplace transform formulae for\texorpdfstring{\\}{}
the quasi-stationary distribution of the Shiryaev diffusion\texorpdfstring{\\}{}
on an interval\thanks{The effort of A.S.~Polunchenko was partially supported by the Simons Foundation via a Collaboration Grant in Mathematics under Award \#\,304574. The work of A.~Pepelyshev was partially supported by the Russian Foundation for Basic Research under Projects \#\#\,17-01-00267-a and 17-01-00161-a.}
}

\titlerunning{On quasi-stationarity of the Shiryaev diffusion}        

\author{
Aleksey S. Polunchenko
        \and
Andrey Pepelyshev 
}

\authorrunning{A.S. Polunchenko \and A. Pepelyshev} 

\institute{
    A.S. Polunchenko \at
    Department of Mathematical Sciences\\
    State University of New York at Binghamton \\
    Binghamton, New York 13902--6000, USA\\
    Tel.: +1-607-777-6906\\
    Fax: +1-607-777-2450\\
    \email{aleksey@binghamton.edu}
\and
    A. Pepelyshev \at
    Cardiff School of Mathematics, Cardiff, CF24 4AG, UK\\
    and\\
    St. Petersburg State University, St. Petersburg, Russia, 199034\\
    \email{pepelyshevan@cardiff.ac.uk}%
}

\date{Received: date / Accepted: date}

\maketitle

\begin{abstract}
We derive analytic closed-form moment and Laplace transform formulae for the quasi-stationary distribution of the classical Shiryaev diffusion restricted to the interval $[0,A]$ with absorption at a given $A>0$.
\end{abstract}

\section{Introduction}
\label{sec:intro}

This work is an investigation into quasi-stationarity of the classical Shiryaev diffusion restricted to an interval. Specifically, the focus is on the solution $(R_{t}^{r})_{t\ge0}$ of the stochastic differential equation
\begin{equation}\label{eq:Rt_r-def}
dR_{t}^{r}
=
dt+R_{t}^{r}dB_t
\;\;
\text{with}
\;\;
R_{0}^{r}\coloneqq r\ge0
\;\;
\text{fixed},
\end{equation}
where $(B_t)_{t\ge0}$ is standard Brownian motion in the sense that $\EV[dB_t]=0$, $\EV[(dB_t)^2]=dt$, and $B_0=0$. The time-homogeneous Markov process $(R_{t}^{r})_{t\ge0}$ is an important particular version of the so-called generalized Shiryaev process. The latter has been first arrived at and studied by Prof. A.N. Shiryaev---hence the name---in his fundamental work~\cite{Shiryaev:SMD61,Shiryaev:TPA63} on quickest change-point detection. While interest to the Shiryaev process in the context of quickest change-point detection has never weakened (see, e.g.,~\cite{Pollak+Siegmund:B85,Shiryaev:Bachelier2002,Feinberg+Shiryaev:SD2006,Burnaev+etal:TPA2009,Polunchenko:SA2016,Polunchenko:SA2017a,Polunchenko:SA2017b,Polunchenko:TPA2017}), the process has received a great deal of attention in other areas as well, notably in mathematical finance (see, e.g.,~\cite{Geman+Yor:MF1993,Donati-Martin+etal:RMI2001,Linetsky:OR2004}) and in mathematical physics (see, e.g.,~\cite{Monthus+Comtet:JPhIF1994,Comtet+Monthus:JPhA1996}). It has also been considered in the literature on general stochastic processes (see, e.g.,~\cite{Wong:SPMPE1964,Yor:AAP1992,Donati-Martin+etal:RMI2001,Dufresne:AAP2001,Schroder:AAP2003,Peskir:Shiryaev2006,Polunchenko+Sokolov:MCAP2016,Polunchenko+etal:TPA2018}).

The particular version of the Shiryaev process $(R_{t}^{r})_{t\ge0}$ governed by equation~\eqref{eq:Rt_r-def} is of special importance and interest because it is the {\em only} version with probabilistically nontrivial behavior in the limit as $t\to+\infty$, exhibited in spite of the distinct martingale property $\EV[R_{t}^{r}-r-t]=0$ for all $t\ge0$ and $r\ge0$. Moreover, the process is convergent (as $t\to+\infty$) regardless of whether the state space is \begin{inparaenum}[(I)]\item the entire half-line $[0,+\infty)$ with no absorption on the interior; or\label{lst:GSR-QSD-case-I} \item the interval $[0,A]$ with absorption at a given level $A>0$; or\label{lst:GSR-QSD-case-II} \item the shortened half-line $[A,+\infty)$ also with absorption at $A>0$ given.\label{lst:GSR-QSD-case-III} \end{inparaenum} The case of a negative initial value $r$ was touched upon in~\cite{Peskir:Shiryaev2006}. Cases~\pref{lst:GSR-QSD-case-I},~\pref{lst:GSR-QSD-case-II}, and~\pref{lst:GSR-QSD-case-III} have all been considered in the literature, which we now briefly review.

Case~\pref{lst:GSR-QSD-case-I} is the easiest case. The asymptotic (as $t\to+\infty$) distribution of $(R_{t}^{r})_{t\ge0}$ in this case is known as the stationary distribution. Formally, the latter is defined as
\begin{equation}\label{eq:SR-StDist-def}
H(x)
\coloneqq
\lim_{t\to+\infty}\Pr(R_{t}^{r}\le x)
\;\;
\text{and}
\;\;
h(x)
\coloneqq
\dfrac{d}{dx}H(x),
\end{equation}
and it has already been found, e.g., in~\cite{Shiryaev:SMD61,Shiryaev:TPA63,Pollak+Siegmund:B85,Feinberg+Shiryaev:SD2006,Burnaev+etal:TPA2009,Polunchenko+Sokolov:MCAP2016}, to be the momentless distribution
\begin{equation}\label{eq:SR-StDist-answer}
H(x)
=
e^{-\tfrac{2}{x}}\indicator{x\ge0}
\;\;
\text{and}
\;\;
h(x)
=
\dfrac{2}{x^2}e^{-\tfrac{2}{x}}\indicator{x\ge0},
\end{equation}
which is an extreme-value Fr\'{e}chet-type distribution. Exact closed-form formulae for the distribution of $R_{t}^{r}$ for {\em any} given $t\ge0$ and $r\ge0$ can be found, e.g., in~\cite{Linetsky:OR2004,Polunchenko+Sokolov:MCAP2016}.

Cases~\pref{lst:GSR-QSD-case-II} and~\pref{lst:GSR-QSD-case-III} are fundamentally different from and far less understood than case~\pref{lst:GSR-QSD-case-I}, due to absorption at one of the boundaries. The corresponding asymptotic (as $t\to+\infty$) distributions are {\em quasi}-stationary distributions, i.e., stationary but conditional on extended survival. Formally, consider the stopping time
\begin{equation*}
\mathcal{S}_{A}^{r}
\coloneqq
\inf\{t\ge0\colon R_{t}^{r}=A\}
\;\;
\text{such that}
\;\;
\inf\{\varnothing\}=+\infty,
\end{equation*}
where $R_{0}^{r}\coloneqq r\ge0$ and $A>0$ are fixed. The quasi-stationary distribution is defined as
\begin{equation}\label{eq:QSD-def}
Q_{A}(x)
\coloneqq
\lim_{t\to+\infty}\Pr(R_{t}^{r}\le x|\mathcal{S}_{A}^{r}>t)
\;\;
\text{and}
\;\;
q_A(x)
\coloneqq
\dfrac{d}{dx}Q_{A}(x),
\end{equation}
and it does depend on whether $r\in[0,A]$, which is case~\pref{lst:GSR-QSD-case-II}, or $r\in[A,+\infty)$, which is case~\pref{lst:GSR-QSD-case-III}, but the specific value of $r$ inside the state space of choice is irrelevant.

Case~\pref{lst:GSR-QSD-case-III} is arguably the least understood case. To the best of our knowledge, the first attempt to treat this case was made in~\cite[Section~7.8.2]{Collet+etal:Book2013} where the authors proved that not only does the quasi-stationary distribution exist for any $A>0$, but also that there is a whole parametric continuum of quasi-stationary distributions when $A$ is not sufficiently large. Further progress on this case was recently made in~\cite{Polunchenko+etal:TPA2018} where $Q_A(x)$ and $q_A(x)$ were, for the first time, found analytically for any $A>0$. It was also shown in~\cite{Polunchenko+etal:TPA2018} that the quasi-stationary distribution is unique whenever $A\ge A^{*}\approx1.265857361$ where $A^{*}$ is the solution of a certain transcendental equation. While case~\pref{lst:GSR-QSD-case-III} may be the least understood case, the focus of this work is entirely on case~\pref{lst:GSR-QSD-case-II}, which is discussed next along with the motivation.

Case~\pref{lst:GSR-QSD-case-II} is of importance in quickest change-point detection, and in this context, it was investigated in, e.g.,~\cite{Pollak+Siegmund:B85,Burnaev+etal:TPA2009,Polunchenko:SA2017a}. See also, e.g.,~\cite{Pollak+Siegmund:JAP1996,Linetsky:OR2004} and~\cite[Section~7.8.2]{Collet+etal:Book2013}. For example, it is known from~\cite{Pollak+Siegmund:B85,Pollak+Siegmund:JAP1996} that, expectedly, the limit of $Q_A(x)$, defined in~\eqref{eq:QSD-def}, as $A\to+\infty$ is $H(x)$, defined in~\eqref{eq:SR-StDist-def} and given by~\eqref{eq:SR-StDist-answer}; the convergence is from above, and is pointwise, at every $x\in[0,+\infty)$, i.e., at all continuity points of $H(x)$. Moreover, analytic closed-form formulae for $Q_A(x)$ and $q_A(x)$ were recently obtained in~\cite{Polunchenko:SA2017a}, apparently for the first time in the literature; see formulae~\eqref{eq:QSD-pdf-answer} and~\eqref{eq:QSD-cdf-answer} below. To boot, the distribution of $R_{t}^{r}$ conditional on no extinction prior to time $t>0$, for {\em any} given $t>0$ and $r\in[0,A)$ has been derived explicitly as well (see, e.g.,~\cite{Polunchenko:SA2016,Linetsky:OR2004}); this conditional distribution becomes the quasi-stationary distribution in the limit, as $t\to+\infty$. Due to its connection to quickest change-point detection, it is case~\pref{lst:GSR-QSD-case-II} that is of interest to this work, which is also motivated by quickest change-point detection. Notwithstanding all the headway made lately on case~\pref{lst:GSR-QSD-case-II}, gaps do remain, and this work seeks to fill some of these gaps in.

More precisely, the contribution of this work in relation to case~\pref{lst:GSR-QSD-case-II} is two-fold: \begin{inparaenum}[\itshape(a)]\item obtain exact closed-form moment formulae for the quasi-stationary distribution; and subsequently use the moment formulae to \item derive an exact formula (in different forms) for the Laplace transform of the quasi-stationary distribution. \end{inparaenum} The moment formulae are obtained as an extension of the effort made earlier in~\cite{Polunchenko:SA2017a} where the moment sequence was shown to satisfy a certain recurrence whose closed-form solution, at the time, seemed out of reach. This work ``runs that leg'' and solves the recurrence explicitly. This is done in the first half of Section~\ref{sec:formulae}, which is the main section of the present paper. The second half of Section~\ref{sec:formulae} is devoted to the computation of the Laplace transform in two different ways: first using the obtained moment formulae, and then also by solving a certain order-two ordinary differential equation that the Laplace transform of interest can be easily shown (see~\cite{Polunchenko:SA2017a}) to satisfy. Since nearly all of the formulae involve special functions, we conveniently preface Section~\ref{sec:formulae} and the derivations therein with Section~\ref{sec:nomenclature} which introduces the relevant special functions. Lastly, Section~\ref{sec:remarks} wraps up the entire paper with a few concluding remarks.

\section{Notation and nomenclature}
\label{sec:nomenclature}

For convenience we shall adapt the standard notation employed uniformly across mathematical literature. In particular, this applies to a host of special functions we shall deal with throughout the sequel. These functions, in their most common notation, are:
%
\begin{enumerate}
    \setlength{\itemsep}{10pt}
    \setlength{\parskip}{0pt}
    \setlength{\parsep}{0pt}
    \item The Gamma function $\Gamma(z)$, $z\in\mathbb{C}$, frequently also referred to as the extension of the factorial to complex numbers, due to the property $\Gamma(n)=(n-1)!$ exhibited for $n\in\mathbb{N}$. See, e.g.,~\cite[Chapter~1]{Bateman+Erdelyi:Book1953v1}.
    \item The Pochhammer symbol, or the rising factorial, often notated as $(z)_n$ and defined for $z\in\mathbb{C}$ and $n\in\mathbb{N}\cup\{0\}$ as
    \begin{equation*}
        (z)_{n}
        \coloneqq
        \begin{cases}
        1,&\text{for $n=0$};\\
        z(z+1)\cdots(z+n-1),&\text{for $n\in\mathbb{N}$},
        \end{cases}
    \end{equation*}
    and it is of note that $(1)_n=n!$ for any $n\in\mathbb{N}\cup\{0\}$. See, e.g.,~\cite[pp.~16--18]{Srivastava+Karlsson:Book1985}. Also, observe that
    \begin{equation*}
    (z)_{n}
    =
    \dfrac{\Gamma(z+n)}{\Gamma(z)}
    \;
    \text{for}
    \;
    n\in\mathbb{N}\cup\{0\}
    \;
    \text{and}
    \;
    z\in\mathbb{C}\setminus\{0,-1,-2,\ldots\},
    \end{equation*}
    and if $z$ is a negative integer or zero, i.e., if $z=-k$ and $k\in\mathbb{N}\cup\{0\}$, then
    \begin{equation}\label{eq:Pochhammer-negint}
    (-k)_{n}
    =
    \begin{cases}
    \dfrac{(-1)^{n}\,k!}{(k-n)!},&\text{for $n=0,1,\ldots,k$};\\[2mm]
    0,&\text{for $n=k+1,k+2,\ldots$};
    \end{cases}
    \end{equation}
    cf.~\cite[p.~16--17]{Srivastava+Karlsson:Book1985}.
    \item The special case of the generalized hypergeometric function (see, e.g.,~\cite[Chapter~4]{Bateman+Erdelyi:Book1953v1}) with two numeratorial and two denominatorial parameters. The function, denoted as ${}_{2}F_{2}[z]$, is defined via the power series
    \begin{equation}\label{eq:2F2-function-def}
        {}_{2}F_{2}
        \left[
            \setlength{\arraycolsep}{0pt}
            \setlength{\extrarowheight}{2pt}
            \begin{array}{@{} c@{{}{}} @{}}
            a_1,a_2
            \\[1ex]
            b_1,b_2
            \\[2pt]
            \end{array}
            \;\middle|\;
            z
        \right]
        \coloneqq
        \sum_{n=0}^{\infty}\dfrac{(a_1)_n\,(a_2)_n}{(b_1)_n\,(b_2)_n}\,\dfrac{z^{n}}{n!},
    \end{equation}
    where $b_1,b_2\not\in\{0,-1,-2,\ldots\}$ and $\abs{z}<+\infty$. See~\cite[p.~20]{Srivastava+Karlsson:Book1985}. It is of note that when only one of the numeratorial parameters $a_i$, $i=1,2$, is a negative integer or zero, then, in view of~\eqref{eq:Pochhammer-negint}, the power series on the right of~\eqref{eq:2F2-function-def} terminates, thereby turning the function ${}_{2}F_{2}[z]$ into a polynomial in $z$ of degree $-a_i$.
    \item The Whittaker $M$ and $W$ functions, traditionally denoted, respectively, as $M_{a,b}(z)$ and $W_{a,b}(z)$, where $a,b,z\in\mathbb{C}$. These functions were introduced by Whittaker~\cite{Whittaker:BAMS1904} as the fundamental solutions to the Whittaker differential equation. See, e.g.,~\cite{Slater:Book1960,Buchholz:Book1969}.
    \item The modified Bessel functions of the first and second kinds, conventionally denoted, respectively, as $I_{a}(z)$ and $K_{a}(z)$, where $a,z\in\mathbb{C}$; the index $a$ is referred to as the function's order. See~\cite[Chapter~7]{Bateman+Erdelyi:Book1953v2}. These functions form a set of fundamental solutions to the modified Bessel differential equation. The modified Bessel $K$ function is also known as the MacDonald function.
    \item The particular case of the generalized bivariate Kamp\'{e} de F\'{e}riet function
    \begin{equation}\label{eq:KdF-function-def}
        F\mathstrut_{2:0;0}^{0:2;1}
        \left[
            \setlength{\arraycolsep}{0pt}
            \setlength{\extrarowheight}{2pt}
            \begin{array}{@{} c@{{}:{}} c@{;{}} c@{} @{}}
            \linefill & a_1,a_2\;\; & \;\; 1
            \\[1ex]
            b_1,b_2 & \linefill & \linefill
            \\[2pt]
            \end{array}
            \;\middle|\;
            xy,x
        \right]
        \coloneqq
        \sum_{i=0}^{\infty}\sum_{j=0}^{\infty}
        \dfrac{(a_{1})_{i}\,(a_{2})_{i}\,(1)_{j}}{(b_{1})_{i+j}\,(b_{2})_{i+j}}\dfrac{(xy)^{i} x^{j}}{i!j!},
    \end{equation}
    which is well-defined for $b_1,b_2\not\in\{0,-1,-2,\ldots\}$ and $\vert x\vert<+\infty$ and $\vert y\vert<+\infty$. See~\cite[p.~27]{Srivastava+Karlsson:Book1985}. The above $F\mathstrut_{2:0;0}^{0:2;1}[x,y]$ function was introduced in~\cite{Lavoie+Grondin:JMAA1994}, and is slightly more general than the original Kamp\'{e} de F\'{e}riet function proposed by Prof.~J. Kamp\'{e} de F\'{e}riet in~\cite{KampeDeFeriet:ASP1921}.
\end{enumerate}

\section{The formulae and discussion}
\label{sec:formulae}

As was mentioned in the introduction, the quasi-stationary distribution defined in~\eqref{eq:QSD-def} was recently expressed analytically in~\cite{Polunchenko:SA2016} through the Whittaker $W$ function. Specifically, it can be deduced from~\cite[Theorem~3.1]{Polunchenko:SA2016} that if $A>0$ is fixed and $\lambda\equiv\lambda_A>0$ is the smallest (positive) solution of the equation
\begin{equation}\label{eq:lambda-eqn}
W_{1,\tfrac{1}{2}\xi(\lambda)}\left(\dfrac{2}{A}\right)
=
0,
\end{equation}
where
\begin{equation}\label{eq:xi-def}
\xi(\lambda)
\coloneqq
\sqrt{1-8\lambda}
\;\;
\text{so that}
\;\;
\lambda
=
\dfrac{1}{8}\left(1-\big[\xi(\lambda)\big]^2\right),
\end{equation}
then the quasi-stationary probability density function (pdf) is given by
\begin{equation}\label{eq:QSD-pdf-answer}
q_A(x)
=
\dfrac{e^{-\tfrac{1}{x}}\,\dfrac{1}{x}\,W_{1,\tfrac{1}{2}\xi(\lambda)}\left(\dfrac{2}{x}\right)}{e^{-\tfrac{1}{A}}\,W_{0,\tfrac{1}{2}\xi(\lambda)}\left(\dfrac{2}{A}\right)}\indicator{x\in[0,A]},
\end{equation}
and the respective cumulative distribution function (cdf) is given by
\begin{equation}\label{eq:QSD-cdf-answer}
Q_A(x)
=
\begin{cases}
1,&\;\text{if $x\ge A$;}\\[2mm]
\dfrac{e^{-\tfrac{1}{x}}\,W_{0,\tfrac{1}{2}\xi(\lambda)}\left(\dfrac{2}{x}\right)}{e^{-\tfrac{1}{A}}\,W_{0,\tfrac{1}{2}\xi(\lambda)}\left(\dfrac{2}{A}\right)},&\;\text{for $x\in[0,A)$;}\\[8mm]
0,&\;\text{otherwise},
\end{cases}
\end{equation}
and $q_A(x)$ and $Q_A(x)$ are each a smooth function of $x\in[0,A]$ and $A>0$; observe that $q_A(A)=0$ as implied by~\eqref{eq:lambda-eqn}, \eqref{eq:xi-def}, and~\eqref{eq:QSD-pdf-answer}. The smoothness of $q_A(x)$ and $Q_A(x)$ is due to the fact that the Whittaker $W$ function on the right of~\eqref{eq:QSD-pdf-answer} and~\eqref{eq:QSD-cdf-answer} is an analytic function of its argument as well as of each of its two indices.
\begin{remark}\label{rem:xi-symmetry}
The definition~\eqref{eq:xi-def} of $\xi(\lambda)$ can actually be changed to $\xi(\lambda)\coloneqq -\sqrt{1-8\lambda}$ with no effect whatsoever on either equation~\eqref{eq:lambda-eqn}, or formulae~\eqref{eq:QSD-pdf-answer} and~\eqref{eq:QSD-cdf-answer}, i.e., all three are invariant with respect to the sign of $\xi(\lambda)$. This was previously pointed out in~\cite{Polunchenko:SA2017a}, and the reason for this $\xi(\lambda)$-symmetry is because equation~\eqref{eq:lambda-eqn} and formulae~\eqref{eq:QSD-pdf-answer} and~\eqref{eq:QSD-cdf-answer} each have $\xi(\lambda)$ present only as (double) the second index of the corresponding Whittaker $W$ function or functions involved, and the Whittaker $W$ function in general is known (see, e.g.,~\cite[Identity~(19),~p.~19]{Buchholz:Book1969}) to be an even function of its second index, i.e., $W_{a,b}(z)=W_{a,-b}(z)$.
\end{remark}

It is evident that equation~\eqref{eq:lambda-eqn} is a key ingredient of formulae~\eqref{eq:QSD-pdf-answer} and~\eqref{eq:QSD-cdf-answer}, and consequently, of all of the characteristics of the quasi-stationary distribution as well. As a transcendental equation, it can only be solved numerically, although to within any desired accuracy, as was previously done, e.g., in~\cite{Linetsky:OR2004,Polunchenko:SA2016,Polunchenko:SA2017a,Polunchenko:SA2017b}, with the aid of {\it Mathematica} developed by Wolfram Research: {\it Mathematica}'s
special functions capabilities have long proven to be superb. Yet, it is known (see~\cite{Linetsky:OR2004,Polunchenko:SA2016}) that for any fixed $A>0$, the equation has countably many simple solutions $0<\lambda_1<\lambda_2<\lambda_3<\cdots$, such that $\lim_{k\to+\infty}\lambda_k=+\infty$. All of them depend on $A$, but since we are interested only in the smallest one, we shall use either the ``short'' notation $\lambda$, or the more explicit $\lambda_A$ to emphasize the dependence on $A$. Also, it can be concluded from~\cite[p.~136~and~Lemma~3.3]{Polunchenko:SA2016} that $\lambda_A$ is a monotonically decreasing function of $A$ such that $\lim_{A\to+\infty}\lambda_A=0$, and more specifically $\lambda_A=A^{-1}+O(A^{-3/2})$.
\begin{remark}\label{rem:xi-complex-real}
Since $\lambda\equiv\lambda_{A}$ is monotonically decreasing in $A$, and such that $\lim_{A\to+\infty}\lambda_{A}=0$, one can conclude from~\eqref{eq:xi-def} that $\xi(\lambda)$ is either \begin{inparaenum}[\itshape(a)]\item purely imaginary (i.e., $\xi(\lambda)=\mathrm{i}\alpha$ where $\mathrm{i}\coloneqq\sqrt{-1}$ and $\alpha\in\mathbb{R}$) if $A$ is sufficiently small, or \item purely real and between 0 inclusive and 1 exclusive (i.e., $0\le \xi(\lambda)<1$) otherwise\end{inparaenum}. The borderline case is when $\xi(\lambda)=0$, i.e., when $\lambda_{A}=1/8$, and the corresponding critical value of $A$ is the solution $\tilde{A}>0$ of the equation $W_{1,0}(2/\tilde{A})=0$. A basic numerical calculation gives $\tilde{A}\approx10.240465$. Hence, if $A<\tilde{A}\approx10.240465$, then $\lambda_{A}>1/8$ so that $\xi(\lambda)$ is purely imaginary; otherwise, if $A\ge\tilde{A}\approx10.240465$, then $\lambda_{A}\in(0,1/8]$ so that $\xi(\lambda)$ is purely real and such that $\xi(\lambda)\in[0,1)$ with $\lim_{A\to+\infty}\xi(\lambda_{A})=1$.
\end{remark}

The asymptotics $\lambda_A=A^{-1}+O(A^{-3/2})$ was first established (in a more general form) in~\cite{Polunchenko:SA2017a} with the aid of Jensen's inequality applied to ascertain that the variance of the quasi-stationary distribution~\eqref{eq:QSD-pdf-answer}--\eqref{eq:QSD-cdf-answer} is strictly positive. This is an example of potential applications of the quasi-stationary distribution's {\em low-order} moments. We now recover the distribution's {\em entire} moment series.

\subsection{The moment series}
\label{ssec:moment-series}

Let $Z$ be a random variable sampled from the quasi-stationary distribution given by~\eqref{eq:QSD-pdf-answer} and~\eqref{eq:QSD-cdf-answer}. Let $\mathfrak{M}_{n}\coloneqq\EV[Z^{n}]$ denote the $n$-th moment of $Z$ for $n\in\mathbb{N}\cup\{0\}$; it is to be understood that $\mathfrak{M}_{0}\equiv1$ for any $A>0$, and that all other $\mathfrak{M}_{n}$'s actually do depend on $A$. For every fixed $A>0$, the series $\{\mathfrak{M}_{n}\}_{n\ge0}$ can be inferred from~\cite[Theorem~3.2,~p.~136]{Polunchenko:SA2017a} to satisfy the recurrence
\begin{equation}\label{eq:QSD-moments-recurrence}
\left(\dfrac{n(n-1)}{2}+\lambda\right)\mathfrak{M}_{n}+n\,\mathfrak{M}_{n-1}
=
\lambda A^{n},
\;\;
n\in\mathbb{N},
\end{equation}
with $\mathfrak{M}_{0}\equiv 1$; recall that $\lambda\equiv\lambda_{A}$ and $A$ are interconnected via equation~\eqref{eq:lambda-eqn}. While recurrence~\eqref{eq:QSD-moments-recurrence} may seem easy to iterate forward on a computer, a general closed-form expression for $\mathfrak{M}_{n}$ for {\em any} $n\in\mathbb{N}\cup\{0\}$ would be more convenient, especially for analytic purposes. To that end, it was lamented in~\cite{Polunchenko:SA2017a} that although the recurrence is possible to solve explicitly, the solution is too cumbersome. We now show that the solution can be expressed compactly through the hypergeometric function ${}_{2}F_{2}[z]$ defined in~\eqref{eq:2F2-function-def}.
\begin{lemma}
For every $A>0$ fixed, the solution $\{\mathfrak{M}_{n}\}_{n\ge0}$ to the recurrence~\eqref{eq:QSD-moments-recurrence} is given by
\begin{equation}\label{eq:QSD-Mn-answer-2F2}
\mathfrak{M}_{n}
=
\dfrac{2\lambda A^{n}}{n(n-1)+2\lambda}
{}_{2}F_{2}
\left[
    \setlength{\arraycolsep}{0pt}
    \setlength{\extrarowheight}{2pt}
    \begin{array}{@{} c@{{}{}} @{}}
    1,-n
    \\[1ex]
    \dfrac{3}{2}+\dfrac{\xi(\lambda)}{2}-n,\dfrac{3}{2}-\dfrac{\xi(\lambda)}{2}-n
    \\[5pt]
    \end{array}
    \;\middle|\;
    \dfrac{2}{A}
\right],
\;
n\in\mathbb{N}\cup\{0\},
\end{equation}
where $\lambda\equiv\lambda_A\;(>0)$ is determined by~\eqref{eq:lambda-eqn} while $\xi(\lambda)$ is defined in~\eqref{eq:xi-def}; recall also that ${}_{2}F_{2}[z]$ denotes the generalized hypergeometric function~\eqref{eq:2F2-function-def}.
\end{lemma}
\begin{proof}
The idea is to first rewrite~\eqref{eq:QSD-moments-recurrence} equivalently as
\begin{equation*}
\big[n(n+1)+2\lambda\big]\mathfrak{M}_{n+1}+2(n+1)\,\mathfrak{M}_{n}
=
2\lambda A^{n+1},
\end{equation*}
and then substitute $\mathfrak{M}_{n}$ of the form
\begin{equation*}
\mathfrak{M}_{n}
=
\dfrac{2\lambda A^{n}}{n(n-1)+2\lambda}\,m(n,A),
\end{equation*}
where $m(n,A)$ is the new unknown. After some elementary algebra this gives
\begin{equation*}
-(n+1)\dfrac{2}{A}m(n,A)+\big[n(n-1)+2\lambda\big]\big[1-m(n+1,A)\big]
=
0,
\end{equation*}
which can be recognized as a particular case of the contiguous function identity
\begin{equation*}
\begin{split}
(b-a)z\,{}_{2}F_{2}
&
\left[
    \setlength{\arraycolsep}{0pt}
    \setlength{\extrarowheight}{2pt}
    \begin{array}{@{} c@{{}{}} @{}}
    a+1,b+1
    \\[1ex]
    c+1,d+1
    \\[2pt]
    \end{array}
    \;\middle|\;
    z
\right]
+
cd\left({}_{2}F_{2}\left[
    \setlength{\arraycolsep}{0pt}
    \setlength{\extrarowheight}{2pt}
    \begin{array}{@{} c@{{}{}} @{}}
    a,b+1
    \\[1ex]
    c,d
    \\[2pt]
    \end{array}
    \;\middle|\;
    z
\right]
-
{}_{2}F_{2}
\left[
    \setlength{\arraycolsep}{0pt}
    \setlength{\extrarowheight}{2pt}
    \begin{array}{@{} c@{{}{}} @{}}
    a+1,b
    \\[1ex]
    c,d
    \\[2pt]
    \end{array}
    \;\middle|\;
    z
\right]\right)
=
0,
\end{split}
\end{equation*}
that the function ${}_{2}F_{2}[z]$ defined in~\eqref{eq:2F2-function-def} is known to satisfy: it suffices to set
\begin{equation*}
a\coloneqq0,
\;
b\coloneqq-n-1,
\;
c\coloneqq-\dfrac{1}{2}+\dfrac{\xi(\lambda)}{2}-n,
\;
d\coloneqq-\dfrac{1}{2}-\dfrac{\xi(\lambda)}{2}-n,
\;
\text{and}
\;
z\coloneqq\dfrac{2}{A},
\end{equation*}
and observe directly from~\eqref{eq:2F2-function-def} that
\begin{equation*}
{}_{2}F_{2}
\left[
    \setlength{\arraycolsep}{0pt}
    \setlength{\extrarowheight}{2pt}
    \begin{array}{@{} c@{{}{}} @{}}
    0,a_2
    \\[1ex]
    b_1,b_2
    \\[2pt]
    \end{array}
    \;\middle|\;
    z
\right]
=
1,
\end{equation*}
for any appropriate $a_2$, $b_1$ and $b_2$.
\qed
\end{proof}

It is clear that the obtained formula~\eqref{eq:QSD-Mn-answer-2F2} is symmetric with respect to $\xi(\lambda)$, as it should be, by Remark~\ref{rem:xi-symmetry}. More importantly, since one of the numeratorial parameters of the function ${}_{2}F_{2}[z]$ on the right of~\eqref{eq:QSD-Mn-answer-2F2} is from the set $\{0,-1,-2,-3,\ldots\}$, the power series buried inside the generalized hypergeometric function terminates, so that $\mathfrak{M}_{n}$ ends up being a polynomial of degree $n$ in $A$. However, the coefficients of the polynomial do depend on $\lambda\equiv\lambda_A$, and since the latter is connected to $A$ via the transcendental equation~\eqref{eq:lambda-eqn}, the actual nature of dependence of $\mathfrak{M}_{n}$ on $A$ is more complicated than polynomial. Specifically, from~\eqref{eq:Pochhammer-negint},~\eqref{eq:2F2-function-def}, and the identity
\begin{equation*}
(z)_{n-k}
=
\dfrac{(-1)^{k}\,(z)_{n}}{(1-z-n)_{k}}
,
\;\;
k=0,1,2,\ldots,n,
\end{equation*}
as given, e.g., by~\cite[Formula~(10),~p.~17]{Srivastava+Karlsson:Book1985}, we readily obtain
\begin{equation*}
\begin{split}
{}_{2}F_{2}
\left[
    \setlength{\arraycolsep}{0pt}
    \setlength{\extrarowheight}{2pt}
    \begin{array}{@{} c@{{}{}} @{}}
    1,-n
    \\[1ex]
    a-n,b-n
    \\[5pt]
    \end{array}
    \;\middle|\;
    z
\right]
&=
\sum_{k=0}^{n}
\dfrac{(1)_{k}\,(-n)_{k}}{(a-n)_{k}\,(b-n)_{k}}\dfrac{z^{k}}{k!}
\\
&
=
\sum_{k=0}^{n}
\dfrac{(-n)_{k}}{(a-n)_{k}\,(b-n)_{k}} z^{k}
\\
&
=
n!
\sum_{k=0}^{n}
\dfrac{1}{(a-n)_{k}\,(b-n)_{k}}\dfrac{(-1)^{k}}{(n-k)!}\,z^{k}
\\
&
=
n!
\sum_{k=0}^{n}
(-1)^{2k}\dfrac{(1-a)_{n-k}\,(1-b)_{n-k}}{(1-a)_{n}\,(1-b)_{n}}\dfrac{(-z)^{k}}{(n-k)!}
\\
&
=
\dfrac{n!\,(-z)^{n}}{(1-a)_{n}\,(1-b)_{n}}
\sum_{k=0}^{n}(1-a)_{k}\,(1-b)_{k}\dfrac{(-z)^{-k}}{k!},
\end{split}
\end{equation*}
whence
\begin{equation}\label{eq:Mn-2F2-sum-form}
\begin{split}
{}_{2}F_{2}
&
\left[
    \setlength{\arraycolsep}{0pt}
    \setlength{\extrarowheight}{2pt}
    \begin{array}{@{} c@{{}{}} @{}}
    1,-n
    \\[1ex]
    \dfrac{3}{2}+\dfrac{\xi(\lambda)}{2}-n,\dfrac{3}{2}-\dfrac{\xi(\lambda)}{2}-n
    \\[5pt]
    \end{array}
    \;\middle|\;
    \dfrac{2}{A}
\right]
=
\\
&\qquad\qquad
=
\dfrac{(-2)^{n}\,n!\,A^{-n}}{\left(-\dfrac{1}{2}+\dfrac{\xi(\lambda)}{2}\right)_{n}\left(-\dfrac{1}{2}-\dfrac{\xi(\lambda)}{2}\right)_{n}}
\times
\\
&\qquad\qquad\qquad
\times
\sum_{k=0}^{n}\left(-\dfrac{1}{2}+\dfrac{\xi(\lambda)}{2}\right)_{k}\left(-\dfrac{1}{2}-\dfrac{\xi(\lambda)}{2}\right)_{k}\dfrac{1}{k!}\left(-\dfrac{A}{2}\right)^{k},
\end{split}
\end{equation}
and subsequently, in view of~\eqref{eq:QSD-Mn-answer-2F2}, we finally find
\begin{equation}\label{eq:QSD-Mn-answer-PowerSeries}
\begin{split}
\mathfrak{M}_{n}
&=
\dfrac{(-2)^{n}\,n!}{\left(\dfrac{1}{2}+\dfrac{\xi(\lambda)}{2}\right)_{n}\left(\dfrac{1}{2}-\dfrac{\xi(\lambda)}{2}\right)_{n}}
\times
\\
&\qquad
\times
\sum_{k=0}^{n}\left(-\dfrac{1}{2}+\dfrac{\xi(\lambda)}{2}\right)_{k}\left(-\dfrac{1}{2}-\dfrac{\xi(\lambda)}{2}\right)_{k}\dfrac{1}{k!}\left(-\dfrac{A}{2}\right)^{k},
\;
n\in\mathbb{N}\cup\{0\},
\end{split}
\end{equation}
where again $\lambda\equiv\lambda_A\;(>0)$ is determined by~\eqref{eq:lambda-eqn} and $\xi(\lambda)$ is defined in~\eqref{eq:xi-def}; this formula is also invariant with respect to the sign of $\xi(\lambda)$.

Let us now briefly contrast the two obtained formulae~\eqref{eq:QSD-Mn-answer-2F2} and~\eqref{eq:QSD-Mn-answer-PowerSeries}. To this end, observe first that formula~\eqref{eq:QSD-Mn-answer-PowerSeries} is more explicit than formula~\eqref{eq:QSD-Mn-answer-2F2}: unlike the latter, the former is free of special functions, and can thus provide more insight into the relationship between $\mathfrak{M}_{n}$ and $A$. A better understanding of this relationship can, in turn, shed more light on the relationship between $\lambda\equiv\lambda_A$ and $A$, an important question difficult to answer by direct analysis of the transcendental equation~\eqref{eq:lambda-eqn} connecting the two. For example, from~\eqref{eq:QSD-Mn-answer-PowerSeries} and the trivial observation that $\mathfrak{M}_{n}>0$ for all $n$ we readily obtain
\begin{equation*}
\mathfrak{M}_{1}
=
A-\dfrac{1}{\lambda_{A}}>0
\;\;
\text{and}
\;
\Var[Z]
=
\mathfrak{M}_{2}-\mathfrak{M}_{1}^{2}
=
\dfrac{\lambda_{A}-(A\lambda_{A}-1)^2}{\lambda_{A}^{2}(1+\lambda_{A})}
>
0,
\end{equation*}
whence
\begin{equation}\label{eq:lambda-double-ineq}
\dfrac{1}{A}
<
\lambda_{A}
<
\dfrac{1}{A}+\dfrac{1+\sqrt{4A+1}}{2A^2}
\;\;
\text{for any}
\;\;
A>0,
\end{equation}
so that $\lambda_{A}=A^{-1}+O(A^{-3/2})$; cf.~\cite{Polunchenko:SA2017a}. For applications of this result in quickest change-point detection see~\cite{Polunchenko:TPA2017,Polunchenko:SA2017b}. Similarly, since the quasi-stationary distribution is supported on the interval $[0,A]$, we may further deduce that $\mathfrak{M}_{n}\le A^{i}\, \mathfrak{M}_{n-i}$ for any $i\in\{0,1,\ldots,n\}$ and $n\in\mathbb{N}\cup\{0\}$. For $n=2$ and $i=1$, after some elementary algebra, this leads to the lower-bound
\begin{equation*}
\dfrac{1}{A}+\dfrac{1}{A+A^2}<\lambda_A
\;\;
\text{for any}
\;\;
A>0,
\end{equation*}
which clearly improves the left half of the double inequality~\eqref{eq:lambda-double-ineq}. By ``playing around'' with the moments more, one can tighten up the lower- and upper-bounds for $\lambda_{A}$ even further, although every such improvement will come at the price of increased complexity of the bounds. That said, the bounds will remain fully amenable to numerical evaluation. See~\cite{Polunchenko:SA2017a} for very accurate high-order bounds.

On the other hand, formula~\eqref{eq:QSD-Mn-answer-2F2} is more convenient than formula~\eqref{eq:QSD-Mn-answer-PowerSeries} to implement in software, especially in Wolfram {\it Mathematica} with its excellent special functions capabilities. To illustrate this point, we implemented formula~\eqref{eq:QSD-Mn-answer-2F2} in a {\it Mathematica} script, and used the script to produce Figures~\ref{fig:Mn_mu1_n_vs_A} and~\ref{fig:Mn_mu1_A_vs_n} which show the behavior of $\mathfrak{M}_{n}$ as a function of $A$ with $n$ fixed and as a function of $n$ with $A$ fixed, respectively; note the different ordinate scales in the figures. Figures~\ref{sfig:Mn_mu1_n1_vs_A}--\ref{sfig:Mn_mu1_n10_vs_A} make it clear that if $n$ is fixed, then $\mathfrak{M}_{n}$ is an increasing function of $A$, concave for $n=1$ and convex otherwise. Given the definition of $\mathfrak{M}_{n}$, the increasing nature of its dependence on $A$ is in alignment with one's intuition. The concavity of the $\mathfrak{M}_{n}$-vs-$A$ curve for $n=1$ and its convexity for $n\ge2$ is due to the aforementioned asymptotics $\lambda_A=A^{-1}+O(A^{-3/2})$, implying $\lim_{A\to+\infty}\big(\lambda_{A}\,A\big)=1$ but $\lim_{A\to+\infty}\big(\lambda_{A}^{1+\kappa}A\big)=0$ for any $\kappa>0$; cf.~\cite{Polunchenko:SA2017a,Polunchenko:SA2017b}. The dependence of $\mathfrak{M}_{n}$ on $n$ for a fixed $A$ has its nuances too: as can be seen from Figures~\ref{sfig:Mn_mu1_A1_vs_n}--\ref{sfig:Mn_mu1_A50_vs_n}, if $A$ is sufficiently small (as in around 1 or even less), then $\mathfrak{M}_n$ is a decreasing function of $n$, and otherwise $\mathfrak{M}_{n}$ is an increasing function of $n$. This is essentially because $f(x)\coloneqq a^{x}$ with $a>0$ is an increasing function of $x$ for $a>1$, and is a decreasing function for $a\in(0,1)$. It is also noteworthy that the rate of growth (or, correspondingly, the rate of decay) of $\mathfrak{M}_{n}$ as a function of $n$ with $A$ fixed or as a function of $A$ with $n$ fixed (at 2 or higher) is rather steep: an eye examination of Figures~\ref{sfig:Mn_mu1_n2_vs_A}--\ref{sfig:Mn_mu1_n10_vs_A} and Figures~\ref{sfig:Mn_mu1_A1_vs_n}--\ref{sfig:Mn_mu1_A50_vs_n} suggests that it is at least exponential, and the rate is the higher, the higher the (fixed) value of $n$ or $A$.
\begin{figure}[!ht]
    \centering
    \begin{subfigure}{0.48\textwidth}
        \centering
        \includegraphics[width=\linewidth]{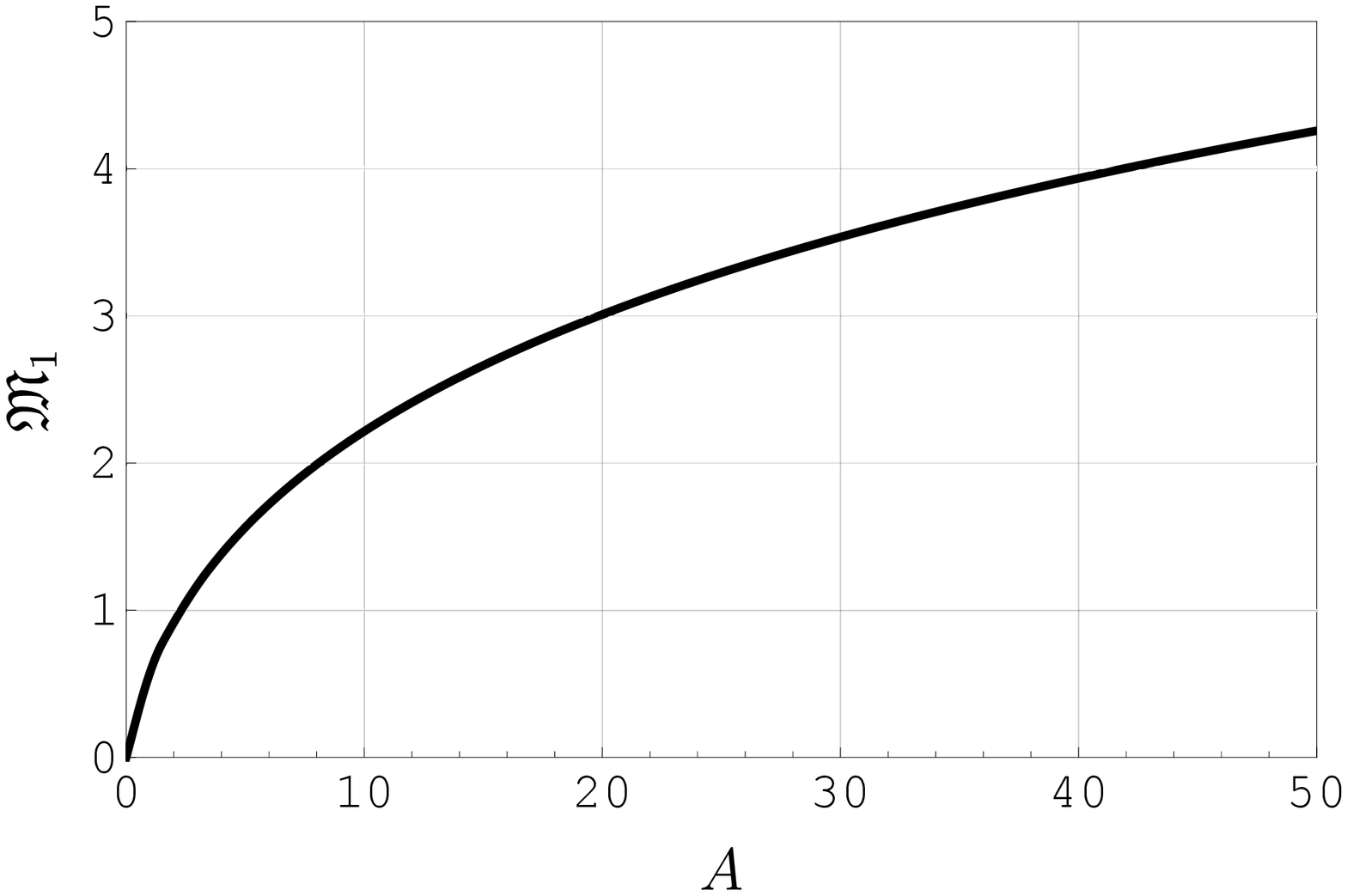}
        \caption{$n=1$.}
        \label{sfig:Mn_mu1_n1_vs_A}
    \end{subfigure}
    \hspace*{\fill}
    \begin{subfigure}{0.48\textwidth}
        \centering
        \includegraphics[width=\linewidth]{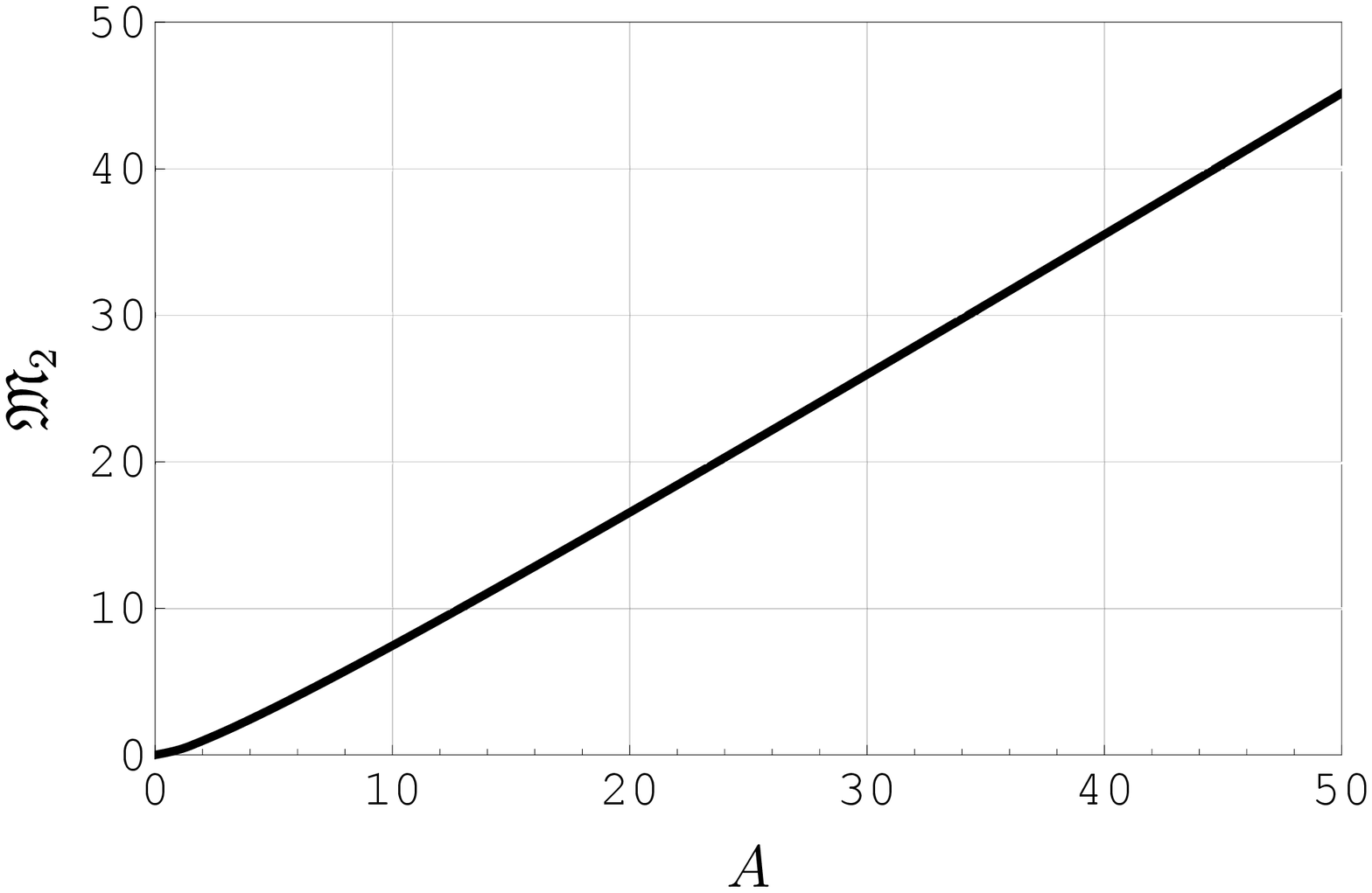}
        \caption{$n=2$.}
        \label{sfig:Mn_mu1_n2_vs_A}
    \end{subfigure}
    \hspace*{\fill}
    \begin{subfigure}{0.48\textwidth}
        \centering
        \includegraphics[width=\linewidth]{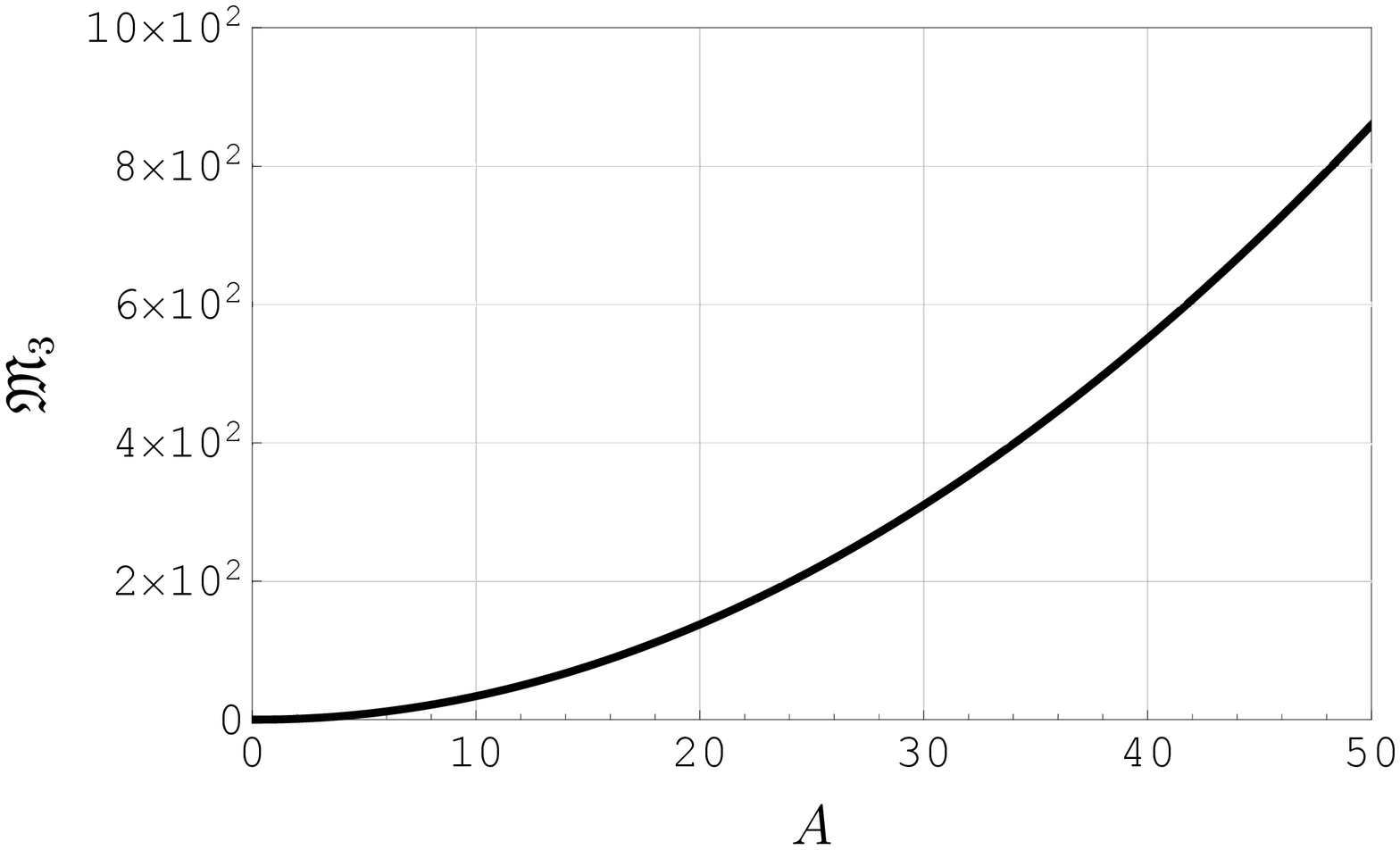}
        \caption{$n=3$.}
        \label{sfig:Mn_mu1_n3_vs_A}
    \end{subfigure}
    \hspace*{\fill}
    \begin{subfigure}{0.48\textwidth}
        \centering
        \includegraphics[width=\linewidth]{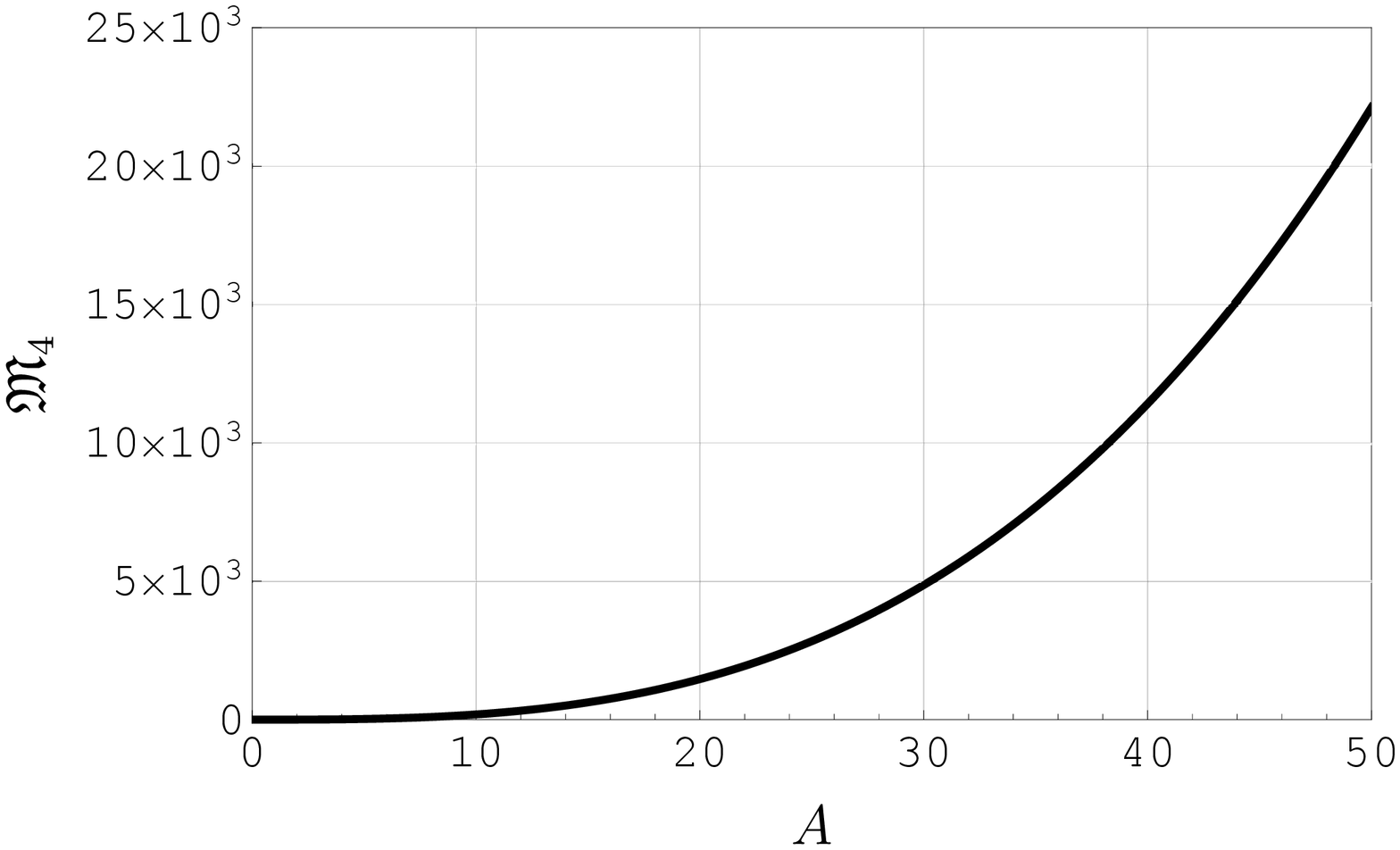}
        \caption{$n=4$.}
        \label{sfig:Mn_mu1_n4_vs_A}
    \end{subfigure}
    \hspace*{\fill}
    \begin{subfigure}{0.48\textwidth}
        \centering
        \includegraphics[width=\linewidth]{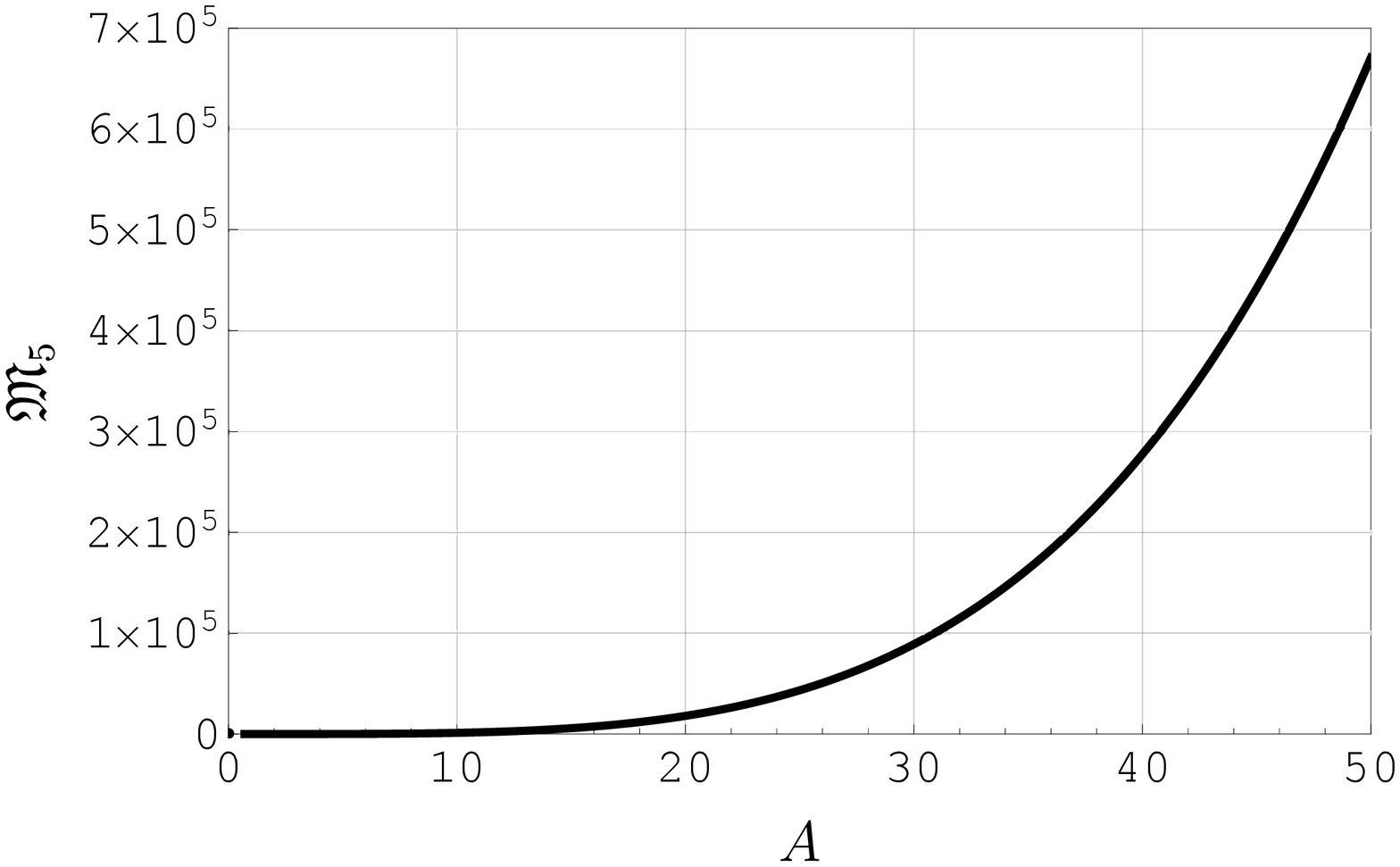}
        \caption{$n=5$.}
        \label{sfig:Mn_mu1_n5_vs_A}
    \end{subfigure}
    \hspace*{\fill}
    \begin{subfigure}{0.48\textwidth}
        \centering
        \includegraphics[width=\linewidth]{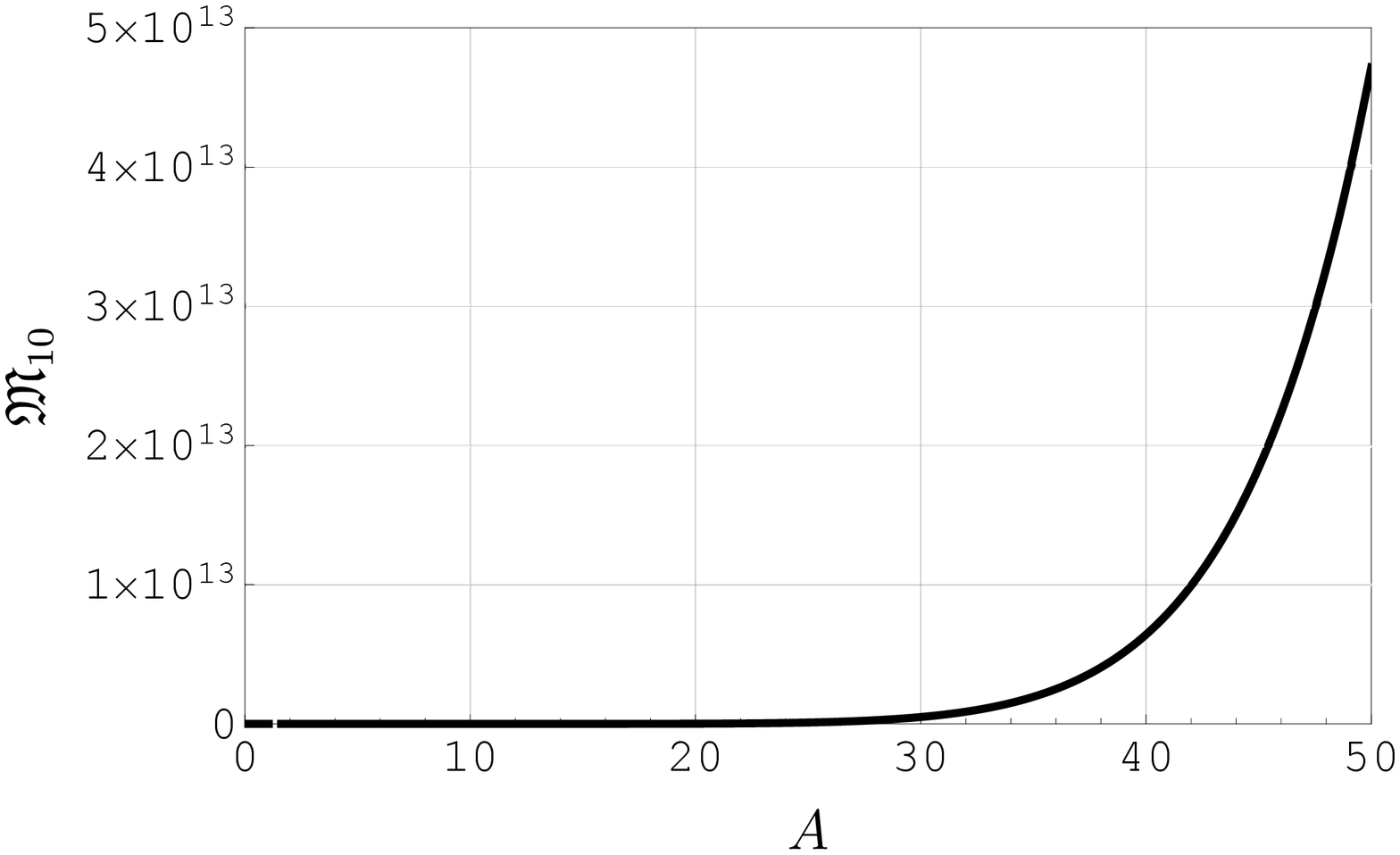}
        \caption{$n=10$.}
        \label{sfig:Mn_mu1_n10_vs_A}
    \end{subfigure}
    \caption{Quasi-stationary distribution's $n$-th moment $\mathfrak{M}_n$ as a function of $A$ for $A\in[0,50]$ and $n\in\{1,2,3,4,5,10\}$.}
    \label{fig:Mn_mu1_n_vs_A}
\end{figure}
\begin{figure}[!ht]
    \centering
    \begin{subfigure}{0.48\textwidth}
        \centering
        \includegraphics[width=\linewidth]{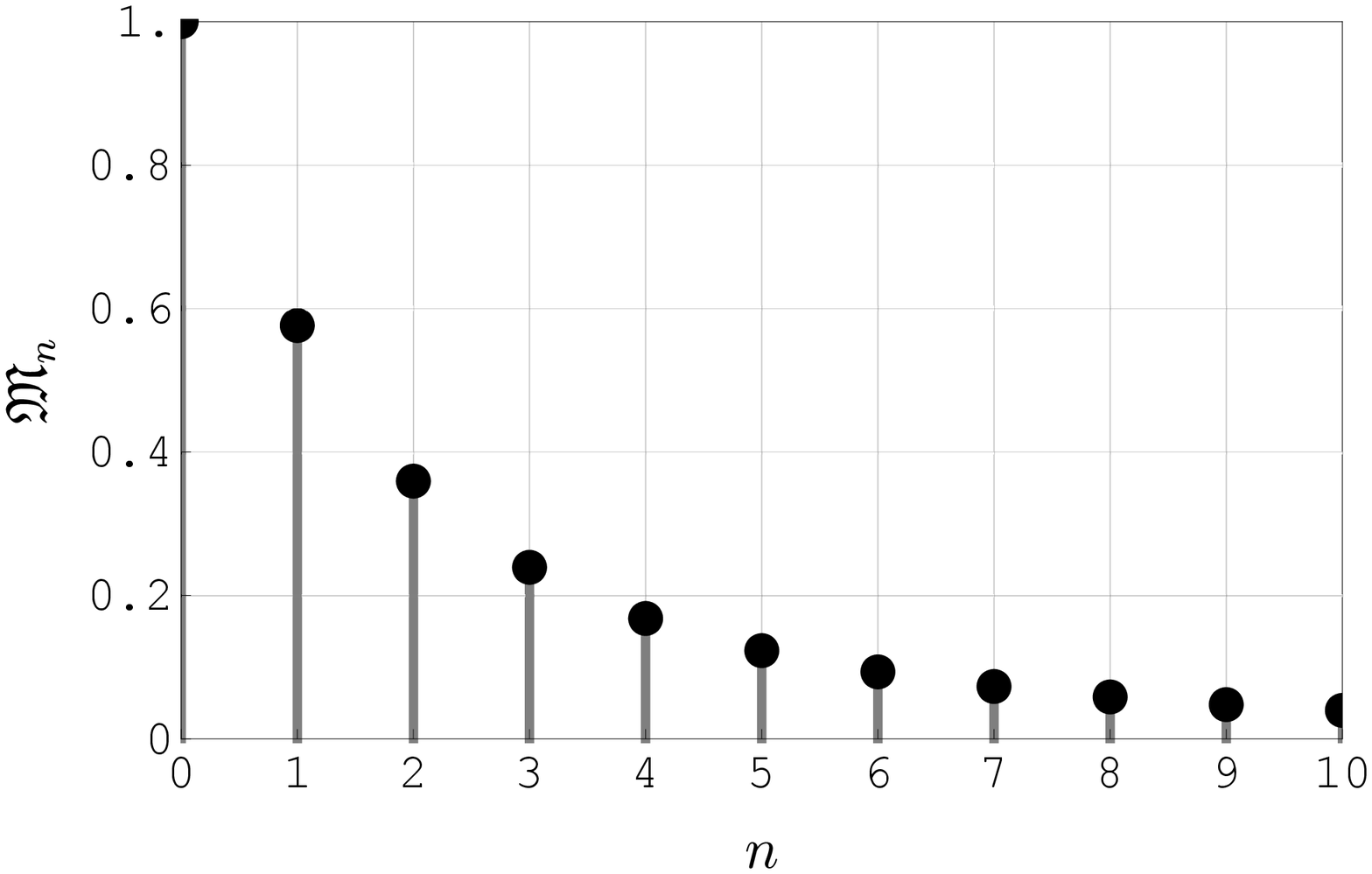}
        \caption{$A=1$.}
        \label{sfig:Mn_mu1_A1_vs_n}
    \end{subfigure}
    \hspace*{\fill}
    \begin{subfigure}{0.48\textwidth}
        \centering
        \includegraphics[width=\linewidth]{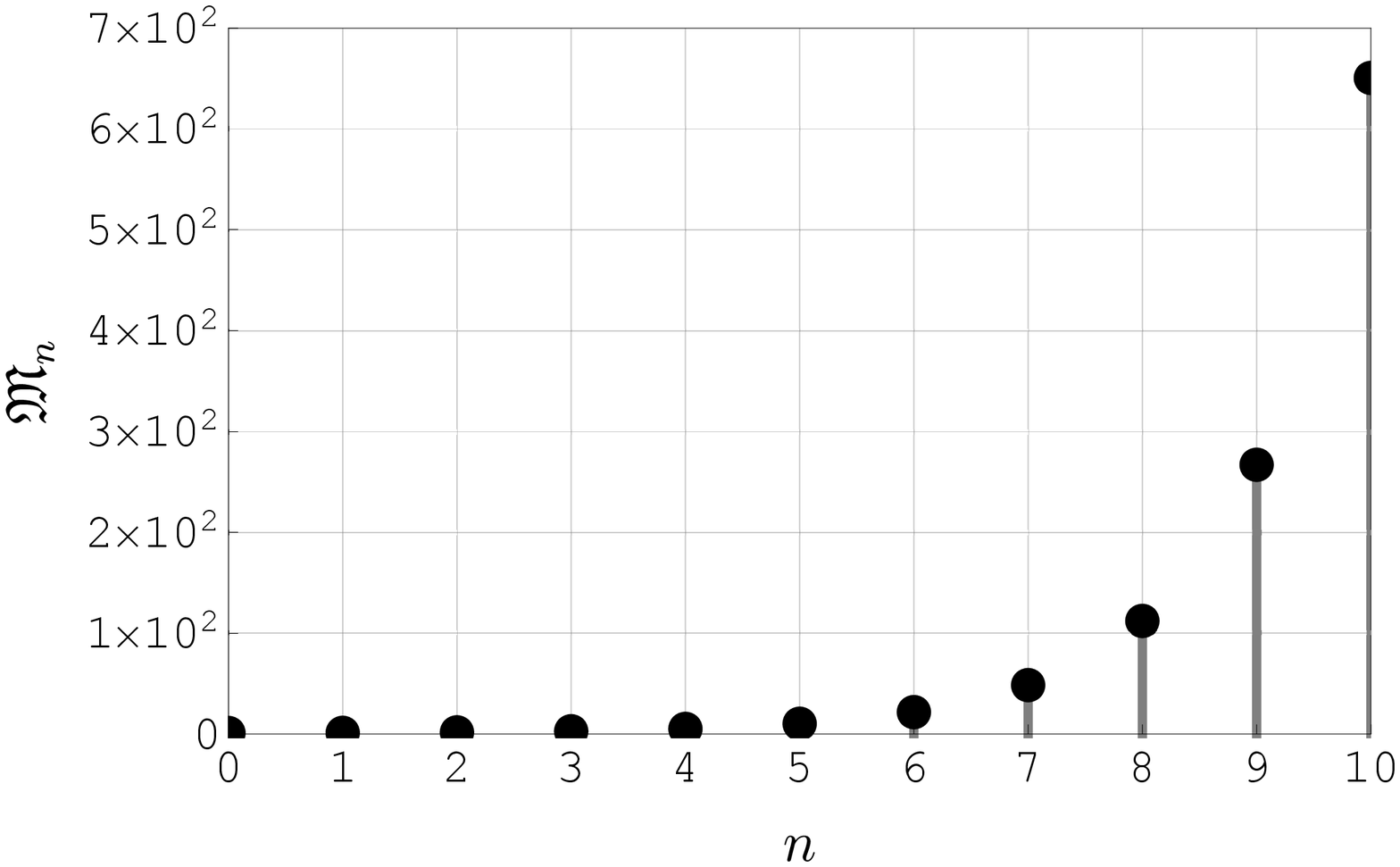}
        \caption{$A=3$.}
        \label{sfig:Mn_mu1_A3_vs_n}
    \end{subfigure}
    \hspace*{\fill}
    \begin{subfigure}{0.48\textwidth}
        \centering
        \includegraphics[width=\linewidth]{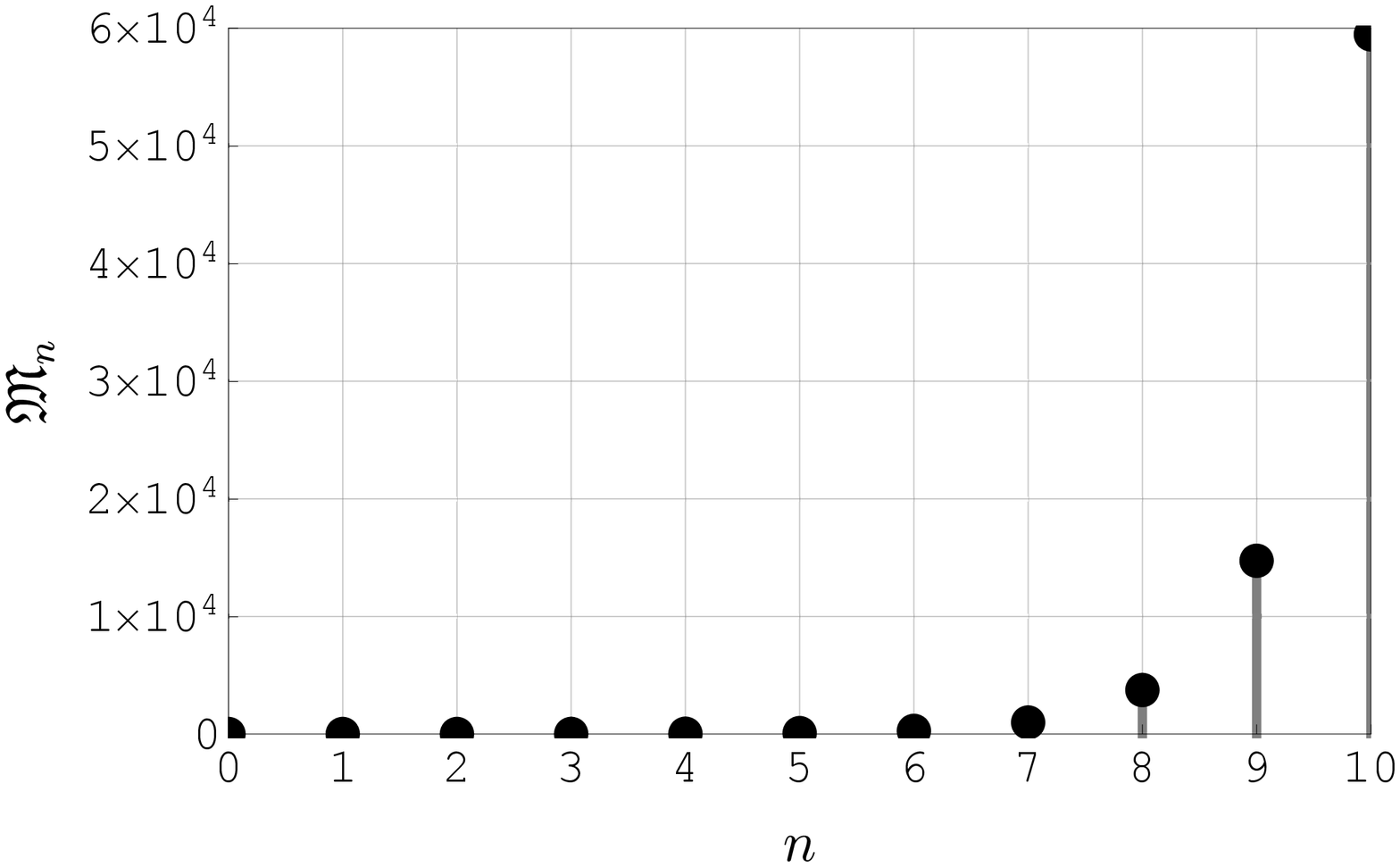}
        \caption{$A=5$.}
        \label{sfig:Mn_mu1_A5_vs_n}
    \end{subfigure}
    \hspace*{\fill}
    \begin{subfigure}{0.48\textwidth}
        \centering
        \includegraphics[width=\linewidth]{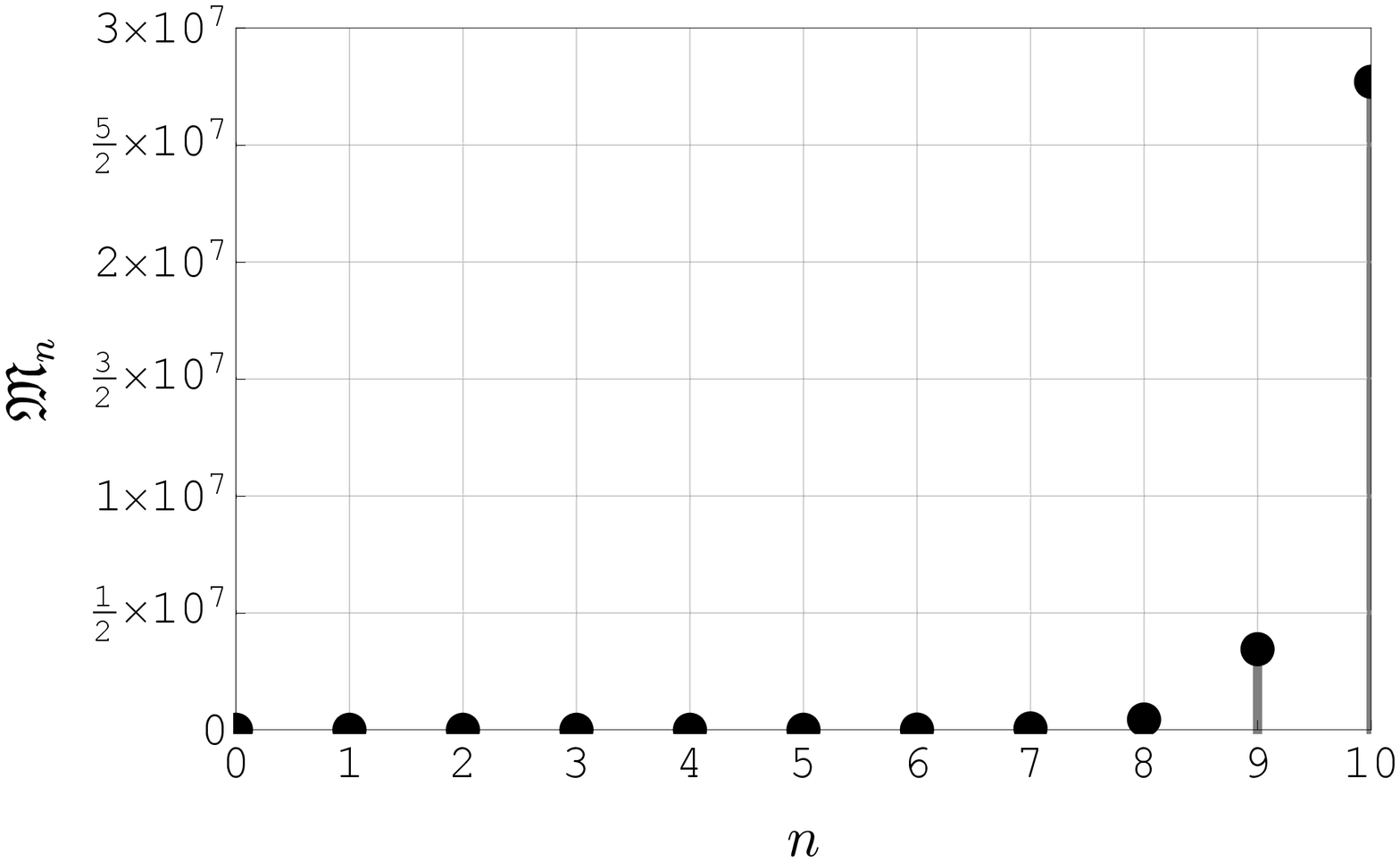}
        \caption{$A=10$.}
        \label{sfig:Mn_mu1_A10_vs_n}
    \end{subfigure}
    \hspace*{\fill}
    \begin{subfigure}{0.48\textwidth}
        \centering
        \includegraphics[width=\linewidth]{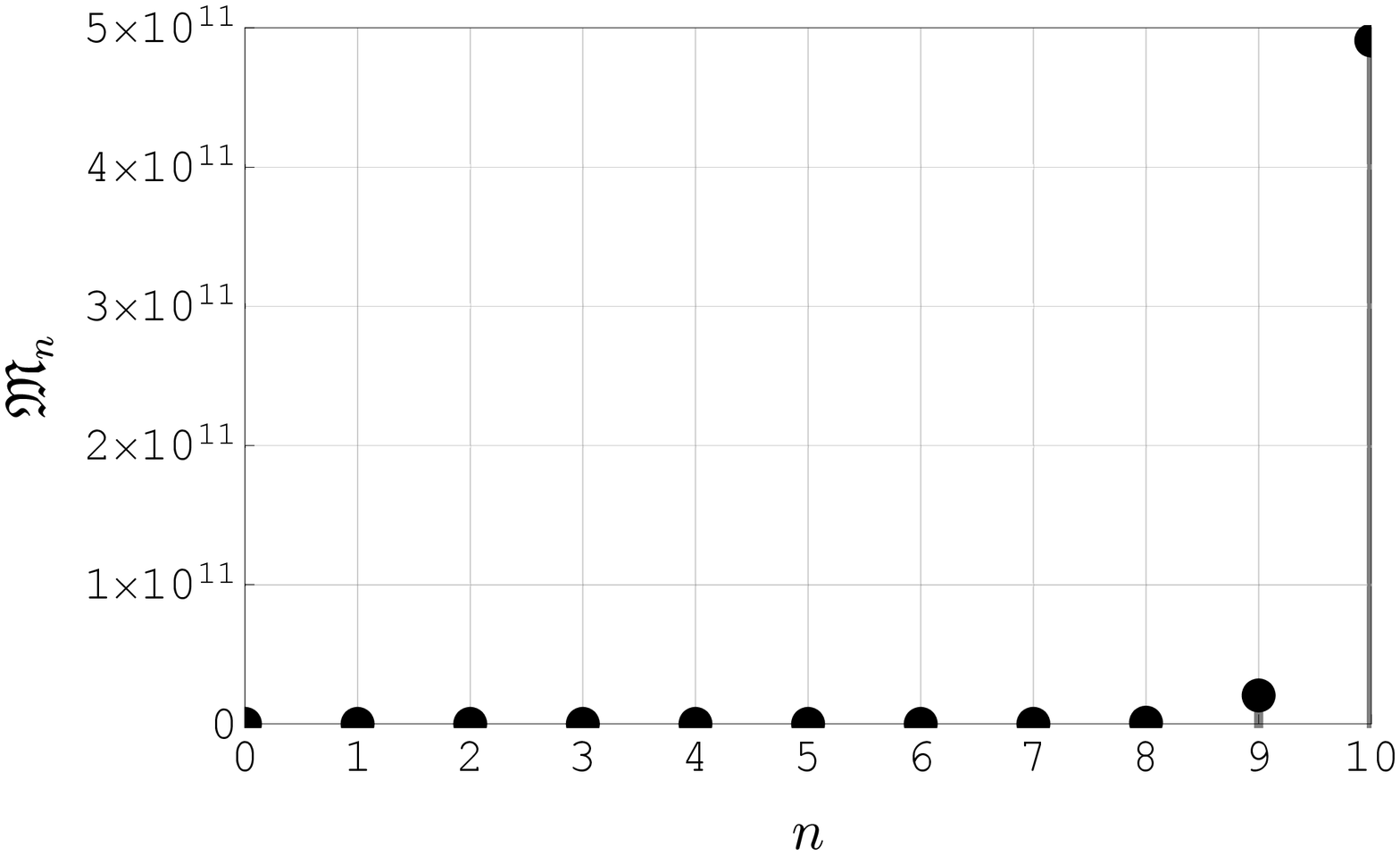}
        \caption{$A=30$.}
        \label{sfig:Mn_mu1_A30_vs_n}
    \end{subfigure}
    \hspace*{\fill}
    \begin{subfigure}{0.48\textwidth}
        \centering
        \includegraphics[width=\linewidth]{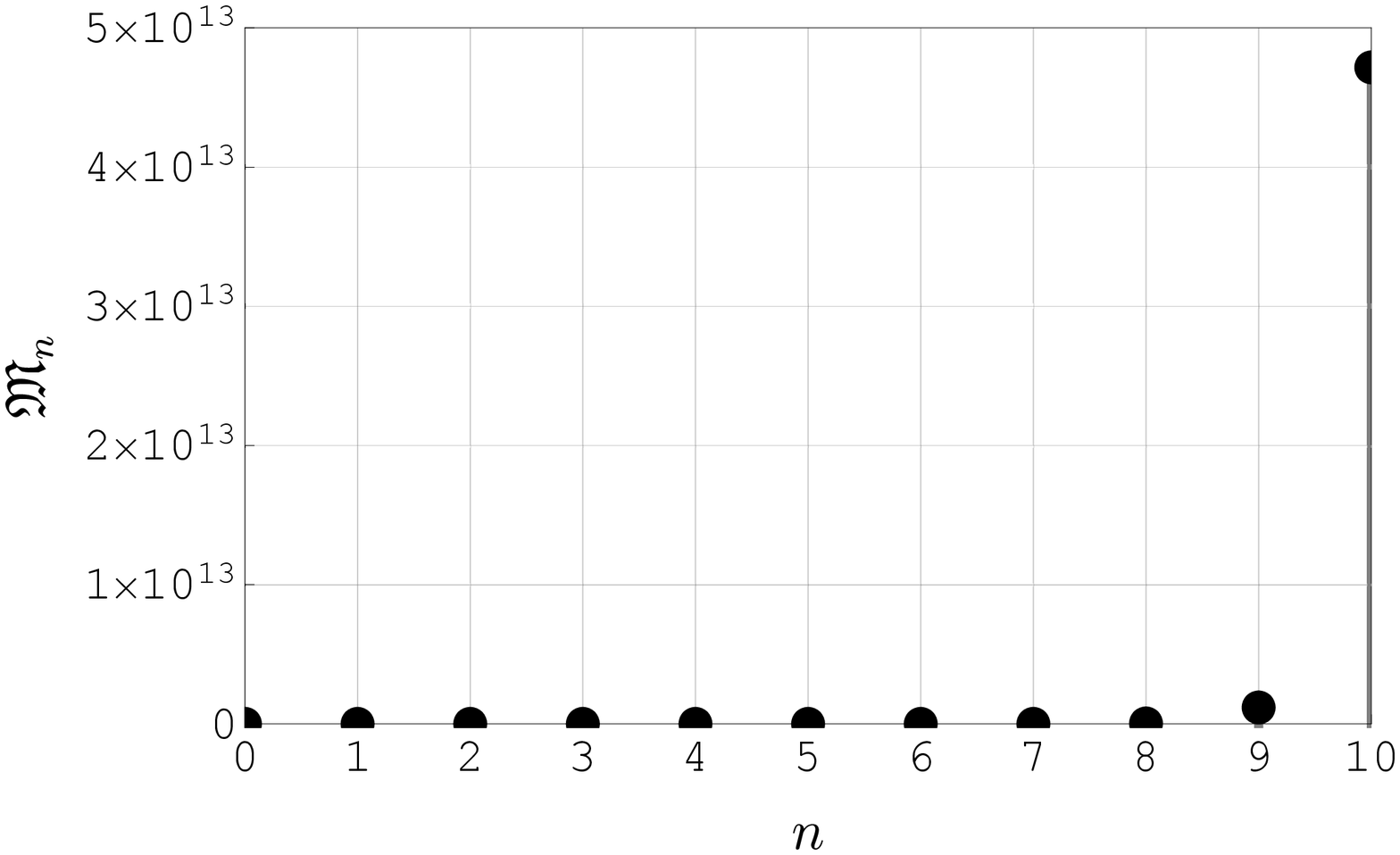}
        \caption{$A=50$.}
        \label{sfig:Mn_mu1_A50_vs_n}
    \end{subfigure}
    \caption{Quasi-stationary distribution's $n$-th moment $\mathfrak{M}_n$ as a function of $n$ for $n\in\{1,2,3,4,5,6,7,8,9,10\}$ and $A\in\{1,3,5,10,30,50\}$.}
    \label{fig:Mn_mu1_A_vs_n}
\end{figure}

However, as we shall see below, should one wish to compute the Laplace transform of the quasi-stationary distribution~\eqref{eq:QSD-pdf-answer}--\eqref{eq:QSD-cdf-answer}, either of the two formulae is instrumental, although one may find formula~\eqref{eq:QSD-Mn-answer-PowerSeries} to be of greater help than formula~\eqref{eq:QSD-Mn-answer-2F2}. The details as well as the actual computation of the Laplace transform are offered in the next subsection.

\subsection{Laplace transform}
\label{ssec:laplace-transform}

We now use the moment formulae obtained above to recover the Laplace transform of the quasi-stationary distribution~\eqref{eq:QSD-def}. Specifically, recall that, for each $A>0$ fixed, the quasi-stationary pdf $q_A(x)$ is given explicitly by~\eqref{eq:QSD-pdf-answer}, and since it is supported on the interval $[0,A]$, its Laplace transform can be defined as the integral
\begin{equation}\label{eq:L-def}
\mathcal{L}_{Q}(s)
\equiv
\mathcal{L}_{Q}\{q_A(x);x\to s\}(s,A)
\coloneqq
\int_{0}^{A} e^{-s x} q_A(x)\,dx,
\;\;
s\ge0,
\end{equation}
and it is connected to the quasi-stationary distribution's moment sequence $\{\mathfrak{M}_{n}\}_{n\ge0}$, given either by~\eqref{eq:QSD-Mn-answer-2F2} or by~\eqref{eq:QSD-Mn-answer-PowerSeries}, via the standard identity
\begin{equation}\label{eq:Mn-L-deriv-s0-formula}
\mathfrak{M}_{n}
=
(-1)^{n}\left.\left[\dfrac{d^{n}}{d s^{n}}\mathcal{L}_{Q}(s)\right]\right|_{s=0},
\end{equation}
leading to the classical power series representation of the Laplace transform
\begin{equation}\label{eq:L-moment-series}
\mathcal{L}_{Q}(s)
=
\sum_{n=0}^{\infty}\dfrac{(-s)^{n}}{n!} \mathfrak{M}_{n},
\end{equation}
which is nothing but the Taylor expansion of $\mathcal{L}_{Q}(s)$ around the origin. It is this expansion, rather than definition~\eqref{eq:L-def}, that we intend to employ shortly to compute $\mathcal{L}_{Q}\{q_A(x);x\to s\}(s,A)$, although with some restrictions on $s$ and $A$. The reason to prefer~\eqref{eq:L-moment-series} along with~\eqref{eq:QSD-Mn-answer-2F2} and~\eqref{eq:QSD-Mn-answer-PowerSeries} over~\eqref{eq:L-def} and~\eqref{eq:QSD-pdf-answer} is the presence of the Whittaker $W$ function on the right of the quasi-stationary pdf formula~\eqref{eq:QSD-pdf-answer}: the Whittaker $W$ function is a special function direct integration of which as in~\eqref{eq:L-def} is unlikely an option, for existing handbooks of special functions appear to offer no suitable integral identities. By contrast, the power series~\eqref{eq:L-moment-series} and the explicit moment formulae~\eqref{eq:QSD-Mn-answer-2F2} and~\eqref{eq:QSD-Mn-answer-PowerSeries} provide a more straightforward way to recover $\mathcal{L}_{Q}(s)$. However, one should keep in mind that the domain of convergence of the series need not be as large as the region of convergence of the integral~\eqref{eq:L-def} defining $\mathcal{L}_{Q}(s)$.
\begin{lemma}
For every $A>0$ fixed and finite, the Laplace transform $\mathcal{L}_{Q}\{q_A(x);x\to s\}(s,A)$ of the quasi-stationary distribution~\eqref{eq:QSD-pdf-answer}--\eqref{eq:QSD-cdf-answer} is given by
\begin{equation}\label{eq:lem-L-KdF-formula1}
\begin{split}
\mathcal{L}_{Q}&\{q_A(x);x\to s\}(s,A)
=\\
&
=
F\mathstrut_{2:0;0}^{0:2;1}
  \left[
   \setlength{\arraycolsep}{0pt}
   \setlength{\extrarowheight}{8pt}
   \begin{array}{@{} c@{{}:{}} c@{;{}} c@{} @{}}
   \linefill & -\dfrac{1}{2}-\dfrac{\xi(\lambda)}{2},-\dfrac{1}{2}+\dfrac{\xi(\lambda)}{2} & \;\; 1
   \\[1ex]
   \dfrac{1}{2}-\dfrac{\xi(\lambda)}{2},\dfrac{1}{2}+\dfrac{\xi(\lambda)}{2} & \linefill & \linefill
   \\[5pt]
   \end{array}
   \;\middle|\;
   -sA,2s
 \right]
 ,
\end{split}
\end{equation}
where $s\in[0,+\infty)$, and $\lambda\equiv\lambda_A\;(>0)$ is determined by~\eqref{eq:lambda-eqn} while $\xi(\lambda)$ is defined in~\eqref{eq:xi-def}; recall also that $F\mathstrut_{2:0;0}^{0:2;1}[x,y]$ denotes the Kamp\'{e} de F\'{e}riet function~\eqref{eq:KdF-function-def}.
\end{lemma}
\begin{proof}
If we tentatively set
\begin{equation*}
a
\coloneqq
\dfrac{1}{2}+\dfrac{\xi(\lambda)}{2}
\;\;
\text{and}
\;\;
b
\coloneqq
\dfrac{1}{2}-\dfrac{\xi(\lambda)}{2}
\end{equation*}
to ease our notation, then together~\eqref{eq:QSD-Mn-answer-PowerSeries},~\eqref{eq:L-moment-series}, and~\eqref{eq:KdF-function-def} can be seen to yield
\begin{equation*}
\begin{split}
\mathcal{L}_{Q}\{q_A(x)&;x\to s\}(s,A)
=
\\
&=
\sum_{n=0}^{\infty}\left\{\dfrac{(2s)^{n}}{(a)_{n}\,(b)_{n}}
\sum_{k=0}^{n}(1-a)_{k}\,(1-b)_{k}\,\dfrac{1}{k!}\left(-\dfrac{A}{2}\right)^{k}\right\}
\\
&=
\sum_{k=0}^{\infty}\left\{(1-a)_{k}\,(1-b)_{k}\,\dfrac{1}{k!}\left(-\dfrac{A}{2}\right)^{k}
\sum_{n=k}^{\infty}\dfrac{(2s)^{n}}{(a)_{n}\,(b)_{n}}\right\}
\\
&=
\sum_{k=0}^{\infty}\sum_{n=0}^{\infty}
\dfrac{(1-a)_{k}\,(1-b)_{k}\,(1)_{n}}{(a)_{n+k}\,(b)_{n+k}}\dfrac{(-sA)^{k} (2s)^{n}}{k!n!}
\\
&=
F\mathstrut_{2:0;0}^{0:2;1}
  \left[
   \setlength{\arraycolsep}{0pt}
   \setlength{\extrarowheight}{2pt}
   \begin{array}{@{} c@{{}:{}} c@{;{}} c@{} @{}}
   \linefill & 1-a,1-b & \;\; 1
   \\[1ex]
   a,b & \linefill & \linefill
   \\[2pt]
   \end{array}
   \;\middle|\;
   -sA,2s
 \right],
\end{split}
\end{equation*}
and the desired result is now apparent.
\qed
\end{proof}

The obtained Laplace transform formula~\eqref{eq:lem-L-KdF-formula1} was arrived at through the transform's power series expansion~\eqref{eq:L-moment-series} and the quasi-stationary distribution's $n$-th moment formula \eqref{eq:QSD-Mn-answer-PowerSeries}. However, since the $n$-th moment also has the alternative but equivalent representation~\eqref{eq:QSD-Mn-answer-2F2}, the latter, too, by virtue of the power series expansion~\eqref{eq:L-moment-series}, can be used to obtain a (different, but equivalent) expression for the Laplace transform.
\begin{lemma}
For every $A>0$ fixed and finite, the Laplace transform $\mathcal{L}_{Q}\{q_A(x);x\to s\}(s,A)$ of the quasi-stationary distribution~\eqref{eq:QSD-pdf-answer}--\eqref{eq:QSD-cdf-answer} is given by
\begin{equation}\label{eq:lem-L-KdF-formula2}
\begin{split}
&\mathcal{L}_{Q}\{q_A(x);x\to s\}(s,A)
=
\dfrac{\lambda}{s}
\times
\\
&
\times
\left(
F\mathstrut_{2:0;0}^{0:2;1}
  \left[
   \setlength{\arraycolsep}{0pt}
   \setlength{\extrarowheight}{8pt}
   \begin{array}{@{} c@{{}:{}} c@{;{}} c@{} @{}}
   \linefill & -\dfrac{1}{2}-\dfrac{\xi(\lambda)}{2},-\dfrac{1}{2}+\dfrac{\xi(\lambda)}{2} & \;\; 1
   \\[1ex]
   -\dfrac{1}{2}-\dfrac{\xi(\lambda)}{2},-\dfrac{1}{2}+\dfrac{\xi(\lambda)}{2} & \linefill & \linefill
   \\[5pt]
   \end{array}
   \;\middle|\;
   -sA,2s
 \right]
\right.
-
\\
&\qquad\qquad\qquad
-
e^{-s A}
\left.
\vphantom{
F\mathstrut_{2:0;0}^{0:2;1}
  \left[
   \setlength{\arraycolsep}{0pt}
   \setlength{\extrarowheight}{8pt}
   \begin{array}{@{} c@{{}:{}} c@{;{}} c@{} @{}}
   \linefill & -\dfrac{1}{2}-\dfrac{\xi(\lambda)}{2},-\dfrac{1}{2}+\dfrac{\xi(\lambda)}{2} & \;\; 1
   \\[1ex]
   -\dfrac{1}{2}-\dfrac{\xi(\lambda)}{2},-\dfrac{1}{2}+\dfrac{\xi(\lambda)}{2} & \linefill & \linefill
   \\[5pt]
   \end{array}
   \;\middle|\;
   -sA,2s
 \right]
}
\right)
,
\end{split}
\end{equation}
where $s\in[0,+\infty)$, and $\lambda\equiv\lambda_A\;(>0)$ is determined by~\eqref{eq:lambda-eqn} while $\xi(\lambda)$ is defined in~\eqref{eq:xi-def}; recall also that $F\mathstrut_{2:0;0}^{0:2;1}[x,y]$ denotes the Kamp\'{e} de F\'{e}riet function~\eqref{eq:KdF-function-def}.
\end{lemma}
\begin{proof}
The idea is to multiply equation~\eqref{eq:QSD-moments-recurrence} through by $(-s)^n/n!$ to obtain
\begin{equation*}
\left(\dfrac{n(n-1)}{2}+\lambda\right)\dfrac{(-s)^{n}}{n!}\mathfrak{M}_{n}-s\dfrac{(-s)^{n-1}}{(n-1)!}\mathfrak{M}_{n-1}
=
\lambda \dfrac{(-sA)^{n}}{n!},
\end{equation*}
which, in conjunction with~\eqref{eq:L-moment-series}, readily gives
\begin{equation*}
\begin{split}
s\mathcal{L}_{Q}\{q_A(x);x\to s\}(s,A)
&=
s\sum_{n=1}^{\infty}\dfrac{(-s)^{n-1}}{(n-1)!}\mathfrak{M}_{n-1}
\\
&=
\sum_{n=1}^{\infty}\left(\dfrac{n(n-1)}{2}+\lambda\right)\dfrac{(-s)^{n}}{n!}\mathfrak{M}_{n}-\lambda\sum_{n=1}^{\infty} \dfrac{(-sA)^{n}}{n!}
\\
&=
\dfrac{1}{2}\left(\,\sum_{n=0}^{\infty}\big[n(n-1)+2\lambda\big]\dfrac{(-s)^{n}}{n!}\mathfrak{M}_{n}-2\lambda\right)-\lambda\left(e^{-sA}-1\right)
\\
&=
\dfrac{1}{2}\sum_{n=0}^{\infty}\big[n(n-1)+2\lambda\big]\dfrac{(-s)^{n}}{n!}\mathfrak{M}_{n}-\lambda e^{-sA},
\end{split}
\end{equation*}
so that if we could now show that
\begin{equation}\label{eq:sumMn-KdF-identity}
\begin{split}
\dfrac{1}{2}&\sum_{n=0}^{\infty}\big[n(n-1)+2\lambda\big]\dfrac{(-s)^{n}}{n!}\mathfrak{M}_{n}
=
\\
&=
\lambda
F\mathstrut_{2:0;0}^{0:2;1}
  \left[
   \setlength{\arraycolsep}{0pt}
   \setlength{\extrarowheight}{8pt}
   \begin{array}{@{} c@{{}:{}} c@{;{}} c@{} @{}}
   \linefill & -\dfrac{1}{2}-\dfrac{\xi(\lambda)}{2},-\dfrac{1}{2}+\dfrac{\xi(\lambda)}{2} & \;\; 1
   \\[1ex]
   -\dfrac{1}{2}-\dfrac{\xi(\lambda)}{2},-\dfrac{1}{2}+\dfrac{\xi(\lambda)}{2} & \linefill & \linefill
   \\[5pt]
   \end{array}
   \;\middle|\;
   -sA,2s
 \right]
,
\end{split}
\end{equation}
then the proof would be complete. To show~\eqref{eq:sumMn-KdF-identity}, introduce
\begin{equation}\label{eq:lem-L-KdF-formula2-sub}
a
\coloneqq
\dfrac{3}{2}+\dfrac{\xi(\lambda)}{2}
\;\;
\text{and}
\;\;
b
\coloneqq
\dfrac{3}{2}-\dfrac{\xi(\lambda)}{2}
\end{equation}
to, again, temporarily ease the notation, and observe from~\eqref{eq:QSD-Mn-answer-2F2} and~\eqref{eq:Mn-2F2-sum-form} that
\begin{equation*}
\begin{split}
\dfrac{1}{2}\sum_{n=0}^{\infty}\big[n(n-1)+2\lambda\big]&\dfrac{(-s)^{n}}{n!}\mathfrak{M}_{n}
=
\\
&=
\lambda\sum_{n=0}^{\infty}
\dfrac{(-sA)^{n}}{n!}
{}_{2}F_{2}
\left[
    \setlength{\arraycolsep}{0pt}
    \setlength{\extrarowheight}{2pt}
    \begin{array}{@{} c@{{}{}} @{}}
    1,-n
    \\[1ex]
    a-n,b-n
    \\[5pt]
    \end{array}
    \;\middle|\;
    \dfrac{2}{A}
\right]
\\
&=
\lambda
\left\{
\sum_{n=0}^{\infty}
\dfrac{(-2s)^{n}}{(1-a)_{n}\,(1-b)_{n}}
\sum_{k=0}^{n}(1-a)_{k}\,(1-b)_{k}\dfrac{1}{k!}\left(-\dfrac{A}{2}\right)^{k}
\right\}
\\
&=
\lambda
\left\{
\sum_{k=0}^{\infty}
(1-a)_{k}\,(1-b)_{k}\dfrac{1}{k!}\left(-\dfrac{A}{2}\right)^{k}
\sum_{n=k}^{\infty}\dfrac{(-2s)^{n}}{(1-a)_{n}\,(1-b)_{n}}
\right\}
\\
&=
\lambda\sum_{k=0}^{\infty}
\sum_{n=0}^{\infty}\dfrac{(1-a)_{k}\,(1-b)_{k}\,(1)_{n}}{(1-a)_{n+k}\,(1-b)_{n+k}}\dfrac{(-sA)^{k}\,(2s)^{n}}{k!n!}
\\
&=
\lambda
F\mathstrut_{2:0;0}^{0:2;1}
  \left[
   \setlength{\arraycolsep}{0pt}
   \setlength{\extrarowheight}{8pt}
   \begin{array}{@{} c@{{}:{}} c@{;{}} c@{} @{}}
   \linefill &1-a,1-b& \;\; 1
   \\[1ex]
   1-a,1-b & \linefill & \linefill
   \\[5pt]
   \end{array}
   \;\middle|\;
   -sA,2s
 \right]
,
\end{split}
\end{equation*}
which, in view of~\eqref{eq:lem-L-KdF-formula2-sub}, can be recognized to be exactly~\eqref{eq:sumMn-KdF-identity}. The proof is now complete.
\qed
\end{proof}

We now return to the point made earlier about the domain of convergence of the series~\eqref{eq:L-moment-series} potentially being narrower than the region of convergence of the integral~\eqref{eq:L-def} defining $\mathcal{L}_{Q}(s)$. This is, in fact, the case, for the obtained Laplace transform formulae~\eqref{eq:lem-L-KdF-formula1} and~\eqref{eq:lem-L-KdF-formula2} both break down in the limit, as either $A\to+\infty$ or $s\to+\infty$. The reason is because the Kamp\'{e} de F\'{e}riet function involved in either formula is well-defined only when both of its two arguments are finite. That said, except for the two limiting cases---one as $A\to+\infty$ and one as $s\to+\infty\,$---formulae~\eqref{eq:lem-L-KdF-formula1} and~\eqref{eq:lem-L-KdF-formula2} are valid.

At this point one may rightly remark that the Kamp\'{e} de F\'{e}riet function in general is a somewhat ``exotic'' special function, although its importance appears to have been well-understood in the literature on mathematical physics. To that end, an interesting question is whether the function  $F\mathstrut_{2:0;0}^{0:2;1}[x,y]$ on the right of formula~\eqref{eq:lem-L-KdF-formula1} permits an alternative expression involving either no special functions at all, or, in the worst case, only ``less exotic'' special functions. While it is very unlikely that our particular function $F\mathstrut_{2:0;0}^{0:2;1}[x,y]$ can be reduced to a form completely free of special functions, it may be possible to express it in terms of fairly widespread modified Bessel functions of the first and second kinds, conventionally denoted as $I_{a}(z)$ and $K_{a}(z)$, respectively. This possibility is indicated by~\cite[Identity~(4.2a),~p.~184]{Miller+Moskowitz:JCAM1998} which states that
\begin{equation}\label{eq:KdF-KI-identity-MillerMoskowitz98}
\begin{split}
&F\mathstrut_{2:0;0}^{0:2;1}
  \left[
   \setlength{\arraycolsep}{0pt}
   \setlength{\extrarowheight}{10pt}
   \begin{array}{@{} c@{{}:{}} c@{;{}} c@{} @{}}
   \linefill & \dfrac{a+b+1}{2},\dfrac{a-b+1}{2} & \;\; 1
   \\[1ex]
   \dfrac{a+b+3}{2},\dfrac{a-b+3}{2} & \linefill & \linefill
   \\[3mm]
   \end{array}
   \;\middle|\;
   x\dfrac{y^2}{4},\dfrac{y^2}{4}
 \right]
=
\\
&\quad
=
\dfrac{(a+b+1)(a-b+1)}{y^{a+1}}\Biggl\{I_{b}(y)\int_{0}^{y}e^{\tfrac{x}{4}u^2}u^{a}K_{b}(u)\,du
-
K_{b}(y)\int_{0}^{y}e^{\tfrac{x}{4}u^2}u^{a}I_{b}(u)\,du\Biggr\},
\end{split}
\end{equation}
valid so long as $\Re(1+a\pm b)>0$; the condition $\Re(1+a\pm b)>0$ is to assure that the near-origin behavior of the modified Bessel $I$ function
\begin{equation}\label{eq:BesselI-small-arg-asymptotics}
I_{b}(z)
\sim
\dfrac{1}{\Gamma(b+1)}\left(\dfrac{z}{2}\right)^{b}
,
\;
\text{as}
\;
z\to0
,
\;\;
\text{provided}
\;\;
b\not\in\{-1,-2,-3,\ldots\},
\end{equation}
as given, e.g., by~\cite[Property~9.6.7,~p.~375]{Abramowitz+Stegun:Handbook1964}, and that of the modified Bessel $K$ function
\begin{equation}\label{eq:BesselK-small-arg-asymptotics}
K_{b}(z)
\sim
\dfrac{1}{2\Gamma(b)}\left(\dfrac{z}{2}\right)^{-b}
,
\;
\text{as}
\;
z\to0,
\;\;
\text{provided}
\;\;
\Re(b)>0,
\end{equation}
as given, e.g., by~\cite[Property~9.6.9,~p.~375]{Abramowitz+Stegun:Handbook1964}, are such that the two integrals on the right of~\eqref{eq:KdF-KI-identity-MillerMoskowitz98}, i.e., the integrals
\begin{equation*}
\int_{0}^{y}e^{\tfrac{x}{4}u^2}u^{a}\,I_{b}(u)\,du
\;\;
\text{and}
\;\;
\int_{0}^{y}e^{\tfrac{x}{4}u^2}u^{a}\,K_{b}(u)\,du,
\end{equation*}
are convergent, for any $y\in[0,+\infty)$; cf.~\cite{Miller+Moskowitz:JFI1991}. Incidentally, the foregoing two integrals are examples of incomplete Weber integrals, which arise in mathematical physics and in certain areas of probability theory; see, e.g.,~\cite{Miller+Moskowitz:JFI1991,Miller+Moskowitz:JCAM1998}.

It is plain to see that the Kamp\'{e} de F\'{e}riet function on the left of identity~\eqref{eq:KdF-KI-identity-MillerMoskowitz98} with $a=-2$ and $b=\xi(\lambda)$ is of precisely the same form as the Kamp\'{e} de F\'{e}riet function on the right of the Laplace transform formula~\eqref{eq:lem-L-KdF-formula1}. However, identity~\eqref{eq:KdF-KI-identity-MillerMoskowitz98} with $a=-2$ and $b=\xi(\lambda)$, which is the case we are interested in, does not hold true. This is due to two reasons. First, the condition $\Re(1+a\pm b)>0$ is false for $a=-2$ and $b=\xi(\lambda)$, because $\xi(\lambda)$, as was explained in Remark~\ref{rem:xi-complex-real}, is either purely imaginary (so that $\Re(b)=0$) or purely real and between 0 inclusive and 1 exclusive (so that $0\le b<1$). The second reason is that, in our case, the parameter $b=\xi(\lambda)$ happens to be connected (and in very specific manner!) to the first argument of the Kamp\'{e} de F\'{e}riet function; the connection is through equation~\eqref{eq:lambda-eqn}. Yet, although not directly applicable in our case,  identity~\eqref{eq:KdF-KI-identity-MillerMoskowitz98} is still of value: observe that its right-hand side resembles the variation of parameters formula for a particular solution to a second-order nonhomogeneous ordinary differential equation. Moreover, this equation is not too difficult to ``reverse engineer''. To this end, it can be deduced from~\cite{Polunchenko:SA2017a} that, for every $A>0$ fixed, the Laplace transform $\mathcal{L}_{Q}(s)\equiv\mathcal{L}_{Q}\{q_A(x);x\to s\}(s,A)$ defined in~\eqref{eq:L-def} is the solution $L(s)\equiv L(s,A)$ of the equation
\begin{equation}\label{eq:L-ode}
\dfrac{s^2}{2}\dfrac{\partial^2}{\partial s^2}L(s)-(s-\lambda)\,L(s)
=
\lambda e^{-sA},
\;\;
s\ge0,
\end{equation}
where recall that $\lambda\equiv\lambda_{A}\;(>0)$ and $A$ are coupled together via equation~\eqref{eq:lambda-eqn}. As we shall see shortly, the right-hand side of identity~\eqref{eq:KdF-KI-identity-MillerMoskowitz98} with $a=-2$ and $b=\xi(\lambda)$ is precisely what the method of variation of parameters yields as a particular solution to the foregoing equation~\eqref{eq:L-ode}. However, this particular solution is not {\em the} solution, because it does not satisfy the appropriate boundary conditions, which are $\lim_{s\to0+}L(s)=1$, $\lim_{s\to+\infty}L(s)=0$, and
\begin{equation}\label{eq:L-ode-bnd-condition-Mn}
\left.\left[\dfrac{d^n}{ds^n}L(s)\right]\right|_{s=0}
=
(-1)^{n}\,\mathfrak{M}_{n},
\;\;
n\in\mathbb{N},
\end{equation}
where $\mathfrak{M}_{n}$ is the $n$-th moment of the quasi-stationary distribution; recall formulae~\eqref{eq:QSD-Mn-answer-2F2} and~\eqref{eq:QSD-Mn-answer-PowerSeries} we established for $\mathfrak{M}_{n}$ in the preceding subsection. The first two of the boundary conditions come from the definition~\eqref{eq:L-def} of the Laplace transform, and the third condition is due to~\eqref{eq:Mn-L-deriv-s0-formula}.

To solve equation~\eqref{eq:L-ode} directly, observe that the change of variables $s\mapsto u\equiv u(s)\coloneqq 2\sqrt{2s}$ and the substitution $L(s)\mapsto L(u)\coloneqq u\, \ell(u)$ together convert the equation into
\begin{equation}\label{eq:ell-ode}
u^2 \dfrac{\partial^2}{\partial u^2}\ell(u)+u \dfrac{\partial}{\partial u}\ell(u)-\left(u^2+\big[\xi(\lambda)\big]^{2}\right)\ell(u)
=
\dfrac{8\lambda}{u}e^{-\tfrac{A}{8}u^2},
\end{equation}
which is a nonhomogeneous version of the modified Bessel equation. Hence, by definition, the two fundamental solutions, $\ell^{(1)}(u)$ and $\ell^{(2)}(u)$, to the homogeneous version of the equation are
\begin{equation*}
\ell^{(1)}(u)
\coloneqq
I_{\xi(\lambda)}(u)
\;\;
\text{and}
\;\;
\ell^{(2)}(u)
\coloneqq
K_{\xi(\lambda)}(u),
\end{equation*}
which can be used to construct a particular solution, $\ell^{\mathrm{(p)}}(u)$, to the nonhomogeneous equation via variation of parameters. Specifically, since the Wronskian between $I_{a}(z)$ and $K_{a}(z)$ is
\begin{equation*}
\mathcal{W}\left\{K_{a}(z),I_{a}(z)\right\}
\coloneqq
K_{a}(z)\dfrac{d}{dz}I_{a}(z)
-
I_{a}(z)\dfrac{d}{dz}K_{a}(z)
=
\dfrac{1}{z},
\end{equation*}
as given, e.g., by~\cite[Formula~9.6.15,~p.~375]{Abramowitz+Stegun:Handbook1964}, the basic variation of parameters formula asserts, after some calculation, that the function
\begin{equation*}
\ell^{\mathrm{(p)}}(u)
\coloneqq
8\lambda\,\Biggl\{I_{\xi(\lambda)}(u)\int^{u}e^{-\tfrac{A}{8}x^2}K_{\xi(\lambda)}(x)\,\dfrac{dx}{x^2}
-
K_{\xi(\lambda)}(u)\int^{u}e^{-\tfrac{A}{8}x^2}I_{\xi(\lambda)}(x)\,\dfrac{dx}{x^2}\Biggr\},
\end{equation*}
when defined, solves the nonhomogeneous equation~\eqref{eq:ell-ode}. Parenthetically, it is worth nothing that, just as the Laplace transform $\mathcal{L}_Q(s)$ should be, by definition~\eqref{eq:L-def} and Remark~\ref{rem:xi-symmetry}, the above function $\ell^{\mathrm{(p)}}(u)$ is, too, an even function of $\xi(\lambda)$, because
\begin{equation}\label{eq:BesselK-BesselI-relation}
K_{a}(z)
=
\pi\,\dfrac{I_{-a}(z)-I_{a}(z)}{2\sin(\pi a)},
\end{equation}
as given, e.g., by~\cite[Identity~9.6.2, p.~375]{Abramowitz+Stegun:Handbook1964}.

The problem now is to understand whether the two indefinite integrals involved in the above function $\ell^{\mathrm{(p)}}(u)$ can be turned into {\em convergent} {\em definite} integrals, so that the result is a well-defined function that still satisfies the nonhomogeneous equation~\eqref{eq:ell-ode}. To that end, it can be gleaned, e.g., from~\cite[p.~99]{Bateman+Erdelyi:Book1953v1}, that
\begin{equation*}
I_{a}(z)
\sim
\dfrac{1}{\sqrt{2\pi z}}\,e^{z}
,
\;
\text{as}
\;
\vert z\vert\to+\infty
,
\;\;
\text{and}
\;\;
K_{a}(z)
\sim
\sqrt{\dfrac{\pi}{2z}}\,e^{-z}
,
\;
\text{as}
\;
\vert z\vert\to+\infty,
\end{equation*}
which, in conjunction with Remark~\ref{rem:xi-complex-real}, enables one to see that the integrals
\begin{equation}\label{eq:IK-ints-z-to-infty}
\int_{z}^{+\infty}e^{-\tfrac{A}{8}x^2}I_{\xi(\lambda)}(x)\,\dfrac{dx}{x^2}
\;\;
\text{and}
\;\;
\int_{z}^{+\infty}e^{-\tfrac{A}{8}x^2}K_{\xi(\lambda)}(x)\,\dfrac{dt}{x^2}
\end{equation}
are both convergent for any $z>0$, but divergent for $z=0$. As a result, one can conclude that the function
\begin{equation*}
\ell^{\mathrm{(p)}}(u)
\coloneqq
8\lambda\,\Biggl\{I_{\xi(\lambda)}(u)\int_{u}^{+\infty}e^{-\tfrac{A}{8}x^2}K_{\xi(\lambda)}(x)\,\dfrac{dx}{x^2}
-
K_{\xi(\lambda)}(u)\int_{u}^{+\infty}e^{-\tfrac{A}{8}x^2}I_{\xi(\lambda)}(x)\,\dfrac{dx}{x^2}\Biggr\},
\end{equation*}
is a well-defined, valid particular solution to equation~\eqref{eq:ell-ode}; note the similarity of $\ell^{\mathrm{(p)}}(u)$ to the right-hand side of identity~\eqref{eq:KdF-KI-identity-MillerMoskowitz98}.

We are now in a position to claim that the general solution to equation~\eqref{eq:L-ode} is of the form
\begin{equation}\label{eq:L-ode-gen-soln}
\begin{split}
L(s)
&=
C_{1}\,2\sqrt{2s}\,I_{\xi(\lambda)}(t)+C_{2}\,2\sqrt{2s}\, K_{\xi(\lambda)}(2\sqrt{2s})
+
\\
&\qquad\qquad
+
8\lambda\,\Biggl\{2\sqrt{2s}\,I_{\xi(\lambda)}(2\sqrt{2s})\int_{2\sqrt{2s}}^{+\infty}e^{-\tfrac{A}{8}t^2}K_{\xi(\lambda)}(t)\,\dfrac{dt}{t^2}
-
\\
&\qquad\qquad\qquad\qquad\qquad
-
2\sqrt{2s}\,K_{\xi(\lambda)}(2\sqrt{2s})\int_{2\sqrt{2s}}^{+\infty}e^{-\tfrac{A}{8}t^2}I_{\xi(\lambda)}(t)\,\dfrac{dt}{t^2}\Biggr\},
\;\;
s\ge0,
\end{split}
\end{equation}
where $C_1$ and $C_2$ are arbitrary constants, each independent of $s$, but possibly dependent on $A$. The only question left to be considered is that of ``pinning down'' the two constants $C_1$ and $C_2$ so as to make the foregoing $L(s)$ satisfy the necessary boundary conditions.

With regard to fitting the boundary conditions, let us first examine the behavior of $L(s)$ given by~\eqref{eq:L-ode-gen-soln} in the limit as $s\to0+$. To that end, from the small-argument asymptotics~\eqref{eq:BesselI-small-arg-asymptotics} of the modified Bessel $I$ function, and the derivative formula
\begin{equation}\label{eq:BesselI-deriv}
\dfrac{d}{dz}I_{a}(z)
=
I_{a+1}(z)+\dfrac{a}{z}I_{a}(z),
\end{equation}
as given, e.g., by~\cite[Identity~8.486.4,~p.~937]{Gradshteyn+Ryzhik:Book2014}, we obtain
\begin{equation}\label{eq:BesselI-intBesselK-lim0}
\begin{split}
\lim_{u\to0+}\Biggl\{u\,&I_{\xi(\lambda)}(u)\int_{u}^{+\infty}e^{-\tfrac{A}{8}x^2}K_{\xi(\lambda)}(x)\,\dfrac{dx}{x^2}\Biggr\}
=
\\
&
=
\lim_{u\to0+}\Biggl\{\left.\Biggl(\int_{u}^{+\infty}e^{-\tfrac{A}{8}x^2}K_{\xi(\lambda)}(x)\,\dfrac{dx}{x^2}\Biggr)\right/\Biggl(\dfrac{1}{u\,I_{\xi(\lambda)}(u)}\Biggr)\Biggr\}
\\
&
\stackrel{(*)}{=}
\lim_{u\to0+}\Biggl\{\left.\Biggl(-e^{-\tfrac{A}{8}u^2}K_{\xi(\lambda)}(u)\,\dfrac{1}{u^2}\Biggr)\right/\Biggl(-\dfrac{1}{\big[u\,I_{\xi(\lambda)}(u)\big]^2}\dfrac{d}{du}\big[u\,I_{\xi(\lambda)}(u)\big]\Biggr)\Biggr\}
\\
&
=
\lim_{u\to0+}\Biggl\{\dfrac{e^{-\tfrac{A}{8}u^2}\big[I_{\xi(\lambda)}(u)\big]^{2}K_{\xi(\lambda)}(u)}{\big[1+\xi(\lambda)\big]\,I_{\xi(\lambda)}(u)+\big[\xi(\lambda)\big]\,u\,I_{\xi(\lambda)+1}(u)}\Biggr\}
\\
&=
\dfrac{1}{2\xi(\lambda)\big[1+\xi(\lambda)\big]}
,
\end{split}
\end{equation}
where equality $(*)$ is due to L'H\^{o}pital's rule, applicable because the corresponding integral of the modified Bessel $K$ function is divergent when the lower limit of integration is zero.

Likewise, from the small-argument asymptotics~\eqref{eq:BesselK-small-arg-asymptotics} of the modified Bessel $K$ function, its symmetry with respect to the order, i.e., $K_{a}(z)=K_{-a}(z)$, trivially implied by~\eqref{eq:BesselK-BesselI-relation}, and the derivative formula
\begin{equation}\label{eq:BesselK-deriv}
\dfrac{d}{dz}K_{a}(z)
=
-K_{a-1}(z)-\dfrac{a}{z}K_{a}(z)
=
-K_{1-a}(z)-\dfrac{a}{z}K_{a}(z),
\end{equation}
as given, e.g., by~\cite[Identity~8.486.12,~p.~938]{Gradshteyn+Ryzhik:Book2014}, we obtain
\begin{equation}\label{eq:BesselK-intBesselI-lim0}
\begin{split}
\lim_{u\to0+}\Biggl\{u\,K_{\xi(\lambda)}(u)&\int_{u}^{+\infty}e^{-\tfrac{A}{8}x^2}I_{\xi(\lambda)}(x)\,\dfrac{dx}{x^2}\Biggr\}
=
\\
&
=
\lim_{u\to0+}\Biggl\{\dfrac{e^{-\tfrac{A}{8}u^2}\big[K_{\xi(\lambda)}(u)\big]^{2}I_{\xi(\lambda)}(u)}{\big[1-\xi(\lambda)\big]\,K_{\xi(\lambda)}(u)-\big[\xi(\lambda)\big]\,u\,K_{1-\xi(\lambda)}(u)}\Biggr\}
\\
&=
\dfrac{1}{2\xi(\lambda)\big[1-\xi(\lambda)\big]}
,
\end{split}
\end{equation}
where we again used L'H\^{o}pital's rule, applicable because the corresponding integral of the modified Bessel $I$ function is divergent when the lower limit of integration is zero.

Next, from the foregoing two limits~\eqref{eq:BesselI-intBesselK-lim0} and~\eqref{eq:BesselK-intBesselI-lim0}, and~\eqref{eq:xi-def} we obtain
\begin{equation*}
\begin{split}
\lim_{s\to0+}
&\Biggl\{2\sqrt{2s}\,K_{\xi(\lambda)}(2\sqrt{2s})\int_{2\sqrt{2s}}^{+\infty}e^{-\tfrac{A}{8}x^2}I_{\xi(\lambda)}(x)\,\dfrac{dx}{x^2}
-
\\
&\qquad\qquad\qquad\qquad
-
2\sqrt{2s}\,I_{\xi(\lambda)}(2\sqrt{2s})\int_{2\sqrt{2s}}^{+\infty}e^{-\tfrac{A}{8}x^2}K_{\xi(\lambda)}(x)\,\dfrac{dt}{x^2}\Biggr\}
=
\\
&\qquad
=
\dfrac{1}{2\xi(\lambda)}\Biggl\{\dfrac{1}{1-\xi(\lambda)}-\dfrac{1}{1+\xi(\lambda)}\Biggr\}
=
\dfrac{1}{8\lambda},
\end{split}
\end{equation*}
whence, recalling again~\eqref{eq:BesselI-small-arg-asymptotics} and~\eqref{eq:BesselK-small-arg-asymptotics}, one finds that $L(s)$ given by~\eqref{eq:L-ode-gen-soln} converges to unity as $s\to0+$, whatever $C_1$ and $C_2$ be. Put another way, it turns out that $\lim_{s\to0+}L(s)=1$, for any choice of $C_1$ and $C_2$.

Let us switch attention to the behavior of $L(s)$ for large values of $s$. To that end, from~\eqref{eq:IK-ints-z-to-infty} and~\eqref{eq:BesselK-deriv} we obtain
\begin{equation*}
\begin{split}
\lim_{u\to+\infty}\Biggl\{u\,I_{\xi(\lambda)}(u)&\int_{u}^{+\infty}e^{-\tfrac{A}{8}x^2}K_{\xi(\lambda)}(x)\,\dfrac{dx}{x^2}\Biggr\}
=
\\
&
=
\lim_{u\to+\infty}\Biggl\{\dfrac{e^{-\tfrac{A}{8}u^2}\big[I_{\xi(\lambda)}(u)\big]^{2}K_{\xi(\lambda)}(u)}{\big[1+\xi(\lambda)\big]\,I_{\xi(\lambda)}(u)+\big[\xi(\lambda)\big]\,u\,I_{\xi(\lambda)+1}(u)}\Biggr\}
=
0
\end{split}
\end{equation*}
and
\begin{equation*}
\lim_{u\to+\infty}\Biggl\{u\,K_{\xi(\lambda)}(u)\int_{u}^{+\infty}e^{-\tfrac{A}{8}x^2}I_{\xi(\lambda)}(x)\,\dfrac{dx}{x^2}\Biggr\}
=
0,
\end{equation*}
so that the limit of $L(s)$ given by~\eqref{eq:L-ode-gen-soln} as $s\to+\infty$ can now be seen to be infinite if $C_1\neq0$, or $0$ if $C_1=0$. Hence, with $C_1$ set to $0$, our function $L(s)$ simplifies down to
\begin{equation}\label{eq:L-ode-gen-soln-C2only}
\begin{split}
L(s)
&=
C_{2}\,2\sqrt{2s}\, K_{\xi(\lambda)}(2\sqrt{2s})
+
\\
&\qquad\qquad
+
8\lambda\,\Biggl\{2\sqrt{2s}\,K_{\xi(\lambda)}(2\sqrt{2s})\int_{2\sqrt{2s}}^{+\infty}e^{-\tfrac{A}{8}t^2}I_{\xi(\lambda)}(t)\,\dfrac{dt}{t^2}
-
\\
&\qquad\qquad\qquad\qquad\qquad
-
2\sqrt{2s}\,I_{\xi(\lambda)}(2\sqrt{2s})\int_{2\sqrt{2s}}^{+\infty}e^{-\tfrac{A}{8}t^2}K_{\xi(\lambda)}(t)\,\dfrac{dt}{t^2}\Biggr\},
\;\;
s\ge0,
\end{split}
\end{equation}
where $C_2$ is still to be found.

To ``pin down'' $C_2$ one may invoke~\eqref{eq:L-ode-bnd-condition-Mn} for any one value of $n\in\mathbb{N}$. The easiest choice is $n=1$, so that, in view of~\eqref{eq:QSD-Mn-answer-PowerSeries}, we obtain
\begin{equation}\label{eq:L-deriv1-s0-M1-bnd-cond}
\left.\left[\dfrac{d}{ds}L(s)\right]\right|_{s=0}
=
-\mathfrak{M}_{1}
=
\dfrac{1}{\lambda}-A,
\end{equation}
which is what we now intend to make $L(s)$ given by~\eqref{eq:L-ode-gen-soln-C2only} satisfy so as to get an equation to subsequently recover $C_2$ from.

To find $dL(s)/ds$, first recall the symmetry $K_{a}(z)=K_{-a}(z)$, and then devise~\eqref{eq:BesselI-deriv} and~\eqref{eq:BesselK-deriv} and integration by parts to establish the indefinite integral identities
\begin{equation*}
\begin{split}
\int^{u}e^{-\tfrac{A}{8}x^2}I_{a}(x)\,\dfrac{dx}{x^{k}}
&=
\dfrac{1}{a+1-k}\Biggl\{u^{1-k}\,e^{-\tfrac{A}{8}u^2}I_{a}(u)
+
\\
&\qquad
+
\dfrac{A}{4}\int^{u}x^2 e^{-\tfrac{A}{8}x^2}I_{a}(x)\,\dfrac{dx}{x^{k}}-
\int^{u}x\, e^{-\tfrac{A}{8}x^2}I_{a+1}(x)\,\dfrac{dx}{x^{k}}\Biggr\},
\end{split}
\end{equation*}
and
\begin{equation*}
\begin{split}
\int^{u}e^{-\tfrac{A}{8}x^2}K_{a}(x)\,\dfrac{dx}{x^{k}}
&=
\dfrac{1}{a-1+k}\Biggl\{u^{1+k}\,e^{-\tfrac{A}{8}t^2}K_{a}(u)
+
\\
&\qquad
+
\dfrac{A}{4}\int^{u}x^2 e^{-\tfrac{A}{8}x^2}K_{a}(x)\,\dfrac{dx}{x^{k}}
+
\int^{u}x\, e^{-\tfrac{A}{8}x^2}K_{1-a}(x)\,\dfrac{dx}{x^{k}}\Biggr\},
\end{split}
\end{equation*}
so that
\begin{equation*}
\begin{split}
\int_{u}^{+\infty}e^{-\tfrac{A}{8}x^2}I_{\xi(\lambda)}(x)\,\dfrac{dx}{x^{2}}
&=
\dfrac{1}{1-\xi(\lambda)} u^{-1}\,e^{-\tfrac{A}{8}u^2}I_{\xi(\lambda)}(u)
+
\\
&\qquad
+
\dfrac{1}{8\lambda}e^{-\tfrac{A}{8}u^2}I_{\xi(\lambda)+1}(u)
-
\\
&\qquad\qquad
-
\dfrac{1}{8\lambda}\left(\dfrac{1}{1+\xi(\lambda)}-\dfrac{A}{4}\right)u\,e^{-\tfrac{A}{8}u^2}I_{\xi(\lambda)}(u)
+
\\
&\qquad\qquad
+
\dfrac{1}{8\lambda}\left(\dfrac{1}{1+\xi(\lambda)}-\dfrac{A}{4}\right)\dfrac{A}{4}\int_{u}^{+\infty}x^2 e^{-\tfrac{A}{8}x^2}I_{\xi(\lambda)}(x)\,dx
-
\\
&\qquad\qquad\qquad
-
\dfrac{1}{8\lambda}\dfrac{1}{1+\xi(\lambda)}\int_{u}^{+\infty}x\, e^{-\tfrac{A}{8}x^2}I_{\xi(\lambda)+1}(x)\,dx,
\end{split}
\end{equation*}
and
\begin{equation*}
\begin{split}
\int_{u}^{+\infty}e^{-\tfrac{A}{8}x^2}K_{\xi(\lambda)}(x)\,\dfrac{dx}{x^{2}}
&=
\dfrac{1}{1+\xi(\lambda)} u^{-1}\,e^{-\tfrac{A}{8}u^2}K_{\xi(\lambda)}(u)
-
\\
&\qquad
-
\dfrac{1}{8\lambda}e^{-\tfrac{A}{8}u^2}K_{1-\xi(\lambda)}(u)
-
\\
&\qquad\qquad
-
\dfrac{1}{8\lambda}\left(\dfrac{1}{1-\xi(\lambda)}-\dfrac{A}{4}\right)u\,e^{-\tfrac{A}{8}u^2}K_{\xi(\lambda)}(u)
+
\\
&\qquad\qquad
+
\dfrac{1}{8\lambda}\left(\dfrac{1}{1-\xi(\lambda)}-\dfrac{A}{4}\right)\dfrac{A}{4}\int_{u}^{+\infty}x^2 e^{-\tfrac{A}{8}x^2}K_{\xi(\lambda)}(x)\,dx
+
\\
&\qquad\qquad\qquad
+
\dfrac{1}{8\lambda}\dfrac{1}{1-\xi(\lambda)}\int_{u}^{+\infty}x\, e^{-\tfrac{A}{8}x^2}K_{1-\xi(\lambda)}(x)\,dx,
\end{split}
\end{equation*}
which is sufficient to compute the limit of $dL(s)/ds$ as $s\to0+$. Specifically, from~\eqref{eq:L-deriv1-s0-M1-bnd-cond}, after quite a bit of algebra involving repeated use of~\eqref{eq:BesselI-small-arg-asymptotics} and~\eqref{eq:BesselK-small-arg-asymptotics}, we find that
\begin{equation*}
C_2
=
\dfrac{1}{1+\xi(\lambda)}\int_{0}^{+\infty}x\,e^{-\tfrac{A}{8}x^2}I_{\xi(\lambda)+1}(x)\,dx
-
\left(\dfrac{1}{1+\xi(\lambda)}-\dfrac{A}{4}\right)\dfrac{A}{4}\int_{0}^{+\infty}x^2e^{-\tfrac{A}{8}x^2}I_{\xi(\lambda)}(x)\,dx,
\end{equation*}
which can be brought to a more explicit form by appealing to~\cite[Identity~6.643.2,~p.~716]{Gradshteyn+Ryzhik:Book2014}, i.e., the definite integral
\begin{equation*}
\int_{0}^{+\infty} x^{\kappa} e^{-c x} I_{2a}(2b\sqrt{x})\dfrac{dx}{\sqrt{x}}
=
\dfrac{\Gamma(\kappa+a+1/2)}{\Gamma(2a+1)}\,e^{\tfrac{b^2}{2c}}\,\dfrac{1}{bc^{\kappa}}\, M_{-\kappa,a}\left(\dfrac{b^2}{c}\right),
\end{equation*}
valid for $\Re(\kappa+a+1/2)>0$; recall that $M_{a,b}(z)$ here denotes the Whittaker $M$ function. The foregoing definite integral immediately gives
\begin{equation*}
\int_{0}^{+\infty}x^2e^{-\tfrac{A}{8}x^2}I_{\xi(\lambda)}(x)\,dx
=
e^{\tfrac{1}{A}}\,\dfrac{\Gamma(1+[\xi(\lambda)+1]/2)}{\Gamma(\xi(\lambda)+1)}\, \dfrac{8}{A}\, M_{-1,\tfrac{1}{2}\xi(\lambda)}\left(\dfrac{2}{A}\right),
\end{equation*}
and
\begin{equation*}
\int_{0}^{+\infty}x\,e^{-\tfrac{A}{8}x^2}I_{\xi(\lambda)+1}(x)\,dx
=
e^{\tfrac{1}{A}}\,\dfrac{\Gamma(1+[\xi(\lambda)+1]/2)}{\Gamma(\xi(\lambda)+2)} \sqrt{\dfrac{8}{A}}\, M_{-\tfrac{1}{2},\tfrac{1}{2}\xi(\lambda)+\tfrac{1}{2}}\left(\dfrac{2}{A}\right),
\end{equation*}
so that
\begin{equation*}
\begin{split}
C_2
&=
e^{\tfrac{1}{A}}\,\dfrac{\xi(\lambda)+1}{\Gamma(\xi(\lambda)+1)}\,\Gamma\left(\dfrac{\xi(\lambda)+1}{2}\right)\Biggl\{\dfrac{1}{\big[1 +\xi(\lambda)\big]^2}\sqrt{\dfrac{2}{A}}\,M_{-\tfrac{1}{2},\tfrac{1}{2}\xi(\lambda)+\tfrac{1}{2}}\left(\dfrac{2}{A}\right)
-
\\
&\qquad\qquad
-\left(\dfrac{1}{1 + \xi(\lambda)} - \dfrac{A}{4}\right)M_{-1, \tfrac{1}{2}\xi(\lambda)}\left(\dfrac{2}{A}\right)\Biggr\},
\end{split}
\end{equation*}
where we also used the factorial property of the Gamma function $\Gamma(z+1)=z\,\Gamma(z)$. Now, from~\cite[Identity~13.4.28,~p.~507]{Abramowitz+Stegun:Handbook1964}, i.e., the identity
\begin{equation*}
2b\,M_{a-\tfrac{1}{2},b-\tfrac{1}{2}}(z)-\sqrt{z}\,M_{a,b}(z)-2b\,M_{a+\tfrac{1}{2},b-\tfrac{1}{2}}(z)
=
0,
\end{equation*}
we find at once that
\begin{equation*}
\sqrt{\dfrac{2}{A}}\,M_{-\tfrac{1}{2},\tfrac{1}{2}\xi(\lambda)+\tfrac{1}{2}}\left(\dfrac{2}{A}\right)
=
\big[1+\xi(\lambda)\big]\Biggl\{M_{-1,\tfrac{1}{2}\xi(\lambda)}\left(\dfrac{2}{A}\right)-M_{0,\tfrac{1}{2}\xi(\lambda)}\left(\dfrac{2}{A}\right)\Biggr\},
\end{equation*}
whence
\begin{equation*}
C_2
=
e^{\tfrac{1}{A}}\,\dfrac{1}{\Gamma(\xi(\lambda)+1)}\,\Gamma\left(\dfrac{\xi(\lambda)+1}{2}\right)\dfrac{A}{2}
\Biggl\{\dfrac{1+\xi(\lambda)}{2}\,M_{-1,\tfrac{1}{2}\xi(\lambda)}\left(\dfrac{2}{A}\right)-\dfrac{2}{A}\,M_{0,\tfrac{1}{2}\xi(\lambda)}\left(\dfrac{2}{A}\right)\Biggr\}
,
\end{equation*}
which is equivalent to
\begin{equation*}
C_2
=
e^{\tfrac{1}{A}}\,\dfrac{1+\xi(\lambda)}{2\,\Gamma(\xi(\lambda)+1)}\,\Gamma\left(\dfrac{\xi(\lambda)+1}{2}\right)\dfrac{A}{2}\,M_{1,\tfrac{1}{2}\xi(\lambda)}\left(\dfrac{2}{A}\right),
\end{equation*}
because of~\cite[Identity~12.4.29,~p.~507]{Abramowitz+Stegun:Handbook1964}, i.e., the recurrence
\begin{equation*}
(1+2b+2a)M_{a+1,b}(z)-(1+2b-2a)M_{a-1,b}(z)
=
2(2a-z)M_{a,b}(z),
\end{equation*}
whereby
\begin{equation*}
\dfrac{1+\xi(\lambda)}{2}\,M_{-1,\tfrac{1}{2}\xi(\lambda)}\left(\dfrac{2}{A}\right)-\dfrac{2}{A}\,M_{0,\tfrac{1}{2}\xi(\lambda)}\left(\dfrac{2}{A}\right)
=
\dfrac{1+\xi(\lambda)}{2}\,M_{1,\tfrac{1}{2}\xi(\lambda)}\left(\dfrac{2}{A}\right).
\end{equation*}

Next, since the Wronskian between $M_{a,b}(z)$ and $W_{a,b}(z)$ is
\begin{equation*}
\mathcal{W}\{M_{a,b}(z),W_{a,b}(z)\}
\coloneqq
M_{a,b}(z)\dfrac{d}{dz}W_{a,b}(z)
-
W_{a,b}(z)\dfrac{d}{dz}M_{a,b}(z)
=
-\dfrac{\Gamma(1+2b)}{\Gamma(1/2+b-a)},
\end{equation*}
as given, e.g., by~\cite[Formula~(2.4.27),~p.~26]{Slater:Book1960}, and because
\begin{equation*}
W_{a-1,b}(z)
=
\dfrac{z-2a}{2(1/2+b-a)(1/2-b-a)}\,W_{a,b}(z)+\dfrac{z}{(1/2+b-a)(1/2-b-a)}\,\dfrac{d}{dz}W_{a,b}(z),
\end{equation*}
as given, e.g., by~\cite[Formula~(2.4.21),~p.~25]{Slater:Book1960}, it follows that
\begin{equation*}
\lambda A\,\Gamma\left(\dfrac{\xi(\lambda)-1}{2}\right)W_{0,\tfrac{1}{2}\xi(\lambda)}\left(\dfrac{2}{A}\right)\,M_{1,\tfrac{1}{2}\xi(\lambda)}\left(\dfrac{2}{A}\right)
=
-\Gamma(\xi(\lambda)+1),
\end{equation*}
where we also appealed to equation~\eqref{eq:lambda-eqn}.

Putting all of the above together, we can finally conclude that
\begin{equation*}
C_2
=
1\left/\Biggl\{e^{\tfrac{1}{A}}\,W_{0,\tfrac{1}{2}\xi(\lambda)}\left(\dfrac{2}{A}\right)\Biggr\}\right.,
\end{equation*}
which is precisely the normalizing factor in the quasi-stationary distribution's formulae~\eqref{eq:QSD-pdf-answer} and~\eqref{eq:QSD-cdf-answer}.


We have now solved the differential equation~\eqref{eq:L-ode} and obtained yet another representation of the Laplace transform $\mathcal{L}_{Q}\{q_A(x);x\to s\}(s,A)$ of the quasi-stationary distribution~\eqref{eq:QSD-def}.
\begin{lemma}
For every $A>0$ fixed, the Laplace transform $\mathcal{L}_{Q}\{q_A(x);x\to s\}(s,A)$ of the quasi-stationary distribution~\eqref{eq:QSD-pdf-answer}--\eqref{eq:QSD-cdf-answer} is given by
\begin{equation}\label{eq:lem-L-KI-formula}
\begin{split}
\mathcal{L}_{Q}\{q_A(x)&;x\to s\}(s,A)
=
\\
&=
2\sqrt{2s}\,K_{\xi(\lambda)}(2\sqrt{2s})\left/\Biggl\{e^{-\tfrac{1}{A}}\,W_{0,\tfrac{1}{2}\xi(\lambda)}\left(\dfrac{2}{A}\right)\Biggr\}\right.
+
\\
&\qquad
+
8\lambda\Biggl\{2\sqrt{2s}\,K_{\xi(\lambda)}(2\sqrt{2s})\int_{2\sqrt{2s}}^{+\infty}e^{-\tfrac{A}{8}x^2}I_{\xi(\lambda)}(x)\,\dfrac{dx}{x^2}
-
\\
&\qquad\qquad\qquad
-2\sqrt{2s}\,I_{\xi(\lambda)}(2\sqrt{2s})\int_{2\sqrt{2s}}^{+\infty}e^{-\tfrac{A}{8}x^2}K_{\xi(\lambda)}(x)\,\dfrac{dx}{x^2}\Biggr\},
\end{split}
\end{equation}
where $s\ge0$, and $\lambda\equiv\lambda_{A}\;(>0)$ is determined by~\eqref{eq:lambda-eqn} while $\xi(\lambda)$ is defined in~\eqref{eq:xi-def}; recall also that $W_{a,b}(z)$ denotes the Whittaker $W$ function, and $I_{a}(z)$ and $K_{a}(z)$ denote the modified Bessel functions of the first and second kinds, respectively.
\end{lemma}

Yet again, from the symmetry of the Whittaker $W$ with respect to the second index, i.e., $W_{a,b}(z)=W_{a,-b}(z)$, one can see that, just like formulae~\eqref{eq:lem-L-KdF-formula1} and~\eqref{eq:lem-L-KdF-formula2} obtained earlier, the new formula~\eqref{eq:lem-L-KI-formula} is also symmetric with respect to $\xi(\lambda)$, as it should be, by definition~\eqref{eq:L-def} and Remark~\ref{rem:xi-symmetry}. However, unlike formulae~\eqref{eq:lem-L-KdF-formula1} and~\eqref{eq:lem-L-KdF-formula2}, the new formula~\eqref{eq:lem-L-KI-formula} is not only free of the Kamp\'{e} de F\'{e}riet function, but more importantly, it is valid even in the limit, as $A\to+\infty$ or as $s\to+\infty$. While the (trivial) limit as $s\to+\infty$ is of little interest, the (nontrivial) limit as $A\to+\infty$ does merit some consideration, especially in the context of quickest change-point detection~\cite{Pollak+Siegmund:B85}. Since
\begin{equation*}
\lim_{A\to+\infty} \Biggl\{e^{-\tfrac{1}{A}}\,W_{0,\tfrac{1}{2}\xi(\lambda)}\left(\dfrac{2}{A}\right)\Biggr\}
=
1,
\end{equation*}
which was observed previously in~\cite[p.~139]{Polunchenko:SA2017a} as an implication of the limits $\lim_{A\to+\infty}\lambda_{A}=0$ and $\lim_{A\to+\infty}\xi(\lambda_{A})=1$, it can be shown directly from~\eqref{eq:lem-L-KI-formula} with the aid of~\eqref{eq:BesselI-small-arg-asymptotics} that
\begin{equation*}
\lim_{A\to+\infty}\mathcal{L}_{Q}\{q_A(x);x\to s\}(s,A)
=
2\sqrt{2s}\,K_{1}(2\sqrt{2s})
\eqqcolon
\mathcal{L}_{H}(s),
\end{equation*}
for every $s\ge0$ fixed. However, in view of~\cite[Identity~(24),~p.~82]{Bateman+Erdelyi:Book1953v2}, i.e., the identity
\begin{equation*}
K_{a}(b z)
=
\dfrac{1}{2}\int_{0}^{+\infty} \left(\dfrac{b}{x}\right)^{a}e^{-\tfrac{z}{2}\bigl(x+\tfrac{b^{2}}{x}\bigr)}\dfrac{dx}{x},
\;\;
\text{valid for $\Re(z)>0$ and $\Re(b^{2} z)>0$},
\end{equation*}
the function $\mathcal{L}_{H}(s)\coloneqq 2\sqrt{2s}\,K_{1}(2\sqrt{2s})$ can be recognized to be the Laplace transform of the Shiryaev diffusion's stationary distribution defined in~\eqref{eq:SR-StDist-def} and given explicitly by~\eqref{eq:SR-StDist-answer}. That is, for every $s\ge0$ fixed, the limit of $\mathcal{L}_{Q}\{q_A(x);x\to s\}(s,A)$ as $A\to+\infty$ is precisely $\mathcal{L}_{H}(s)$, and, therefore, the stationary distribution~\eqref{eq:SR-StDist-answer} is the limit of the quasi-stationary distribution~\eqref{eq:QSD-pdf-answer}--\eqref{eq:QSD-cdf-answer} as $A\to+\infty$. This convergence of distributions (for a more general family of stochastically monotone processes) was previously established by Pollak and Siegmund in~\cite{Pollak+Siegmund:B85,Pollak+Siegmund:JAP1996}, although through an entirely different approach and with no explicit formulae.

We conclude with an admission that, in our derivation of the Laplace transform formula~\eqref{eq:lem-L-KI-formula}, we actually had to ``cut some corners''. Strictly speaking, by Remark~\ref{rem:xi-complex-real}, we should have considered separately three different cases: \begin{inparaenum}[\itshape(1)]\item $A<\tilde{A}\approx10.240465$ so that $\lambda_{A}>1/8$ and $\xi(\lambda)$ is purely imaginary; \item $A=\tilde{A}\approx10.240465$ so that $\lambda_{A}=1/8$ and $\xi(\lambda)=0$; and \item $A>\tilde{A}\approx10.240465$ so that $\lambda_{A}<1/8$ and $\xi(\lambda)$ is purely real and strictly between 0 and 1\end{inparaenum}. However, for lack of space, we only attended to the third case. The reason to distinguish the three cases is because the asymptotics of the modified Bessel $I$ and $K$ functions are highly order-dependent, and, in our specific situation, the order of either function is determined entirely by $\xi(\lambda)$. For example, the limits~\eqref{eq:BesselI-intBesselK-lim0} and~\eqref{eq:BesselK-intBesselI-lim0} are clearly false when $\xi(\lambda)=0$. Nevertheless, the end-result, viz. formula~\eqref{eq:lem-L-KI-formula}, is valid in all three cases.

\section{Concluding remarks}
\label{sec:remarks}

It is generally rare that quasi-stationary distributions and their characteristics lend themselves to explicit analytic evaluation. Furthermore, in the rare cases one {\em can} recover the distribution itself or its characteristics analytically, the result is usually of limited use, for the corresponding formulae, though explicit, are typically rather complex and involve special functions (or, worse yet, {\em exotic} special functions). This work, as a continuation of~\cite{Polunchenko:SA2017a} and a spin-off of~\cite{Polunchenko+etal:TPA2018}, provided an example of a situation where the distribution itself, its Laplace transform as well as the entire moment series are {\em all} obtainable analytically and in closed-form, despite the presence of special functions in all of the calculations. It is our hope that the special functions calculus heavily used in this work will aid further research on stochastic processes, an area where special functions (including those dealt with in the this paper) arise routinely.

\begin{acknowledgements}
The authors would like to thank the two anonymous referees for the careful reading of the manuscript and pertinent comments; the referees' constructive feedback helped substantially improve the quality of this work and shape its final form.
\end{acknowledgements}


%
%

\bibliographystyle{spmpsci}
\bibliography{main,finance,physics,special-functions,stochastic-processes}

\end{document}